\documentclass[a4paper, 11pt]{article} 
\usepackage{amsfonts,amssymb,amsmath}
\usepackage{epsfig,epsf}
\usepackage{graphicx}
\usepackage{subcaption}
\usepackage{a4wide}
\usepackage{amsmath}
\usepackage{mathtools}
\usepackage{bbm} 
\usepackage{amsthm}
\usepackage[toc,page]{appendix}
\usepackage{afterpage}
\usepackage{enumitem}

\usepackage{xcolor}
\usepackage{url}

\usepackage{authblk}
\usepackage{orcidlink}
\usepackage{cleveref}

\allowdisplaybreaks

\usepackage[round]{natbib} 
\bibpunct{(}{)}{;}{a}{}{,} 

\setlist[enumerate,1]{label=\arabic*.,ref=\arabic*}
\setlist[enumerate,2]{label=\alph*.,ref=\arabic{enumi}.\alph*}
\setlist[enumerate,3]{label=\roman*.,ref=\arabic{enumi}.\alph{enumii}.\roman*}

\numberwithin{equation}{section}

\allowdisplaybreaks

\numberwithin{figure}{section}
\newtheorem{thm}{Theorem}[section]
\newtheorem{cor}[thm]{Corollary}

\newtheorem{prop}[thm]{Proposition}

\theoremstyle{definition}

\newtheorem{remark}[thm]{Remark}

\newcommand{\Real}{\mathbb R}

\newcommand{\diag}{\mathrm{diag}}

\newcommand*{\ud}{\mathrm{d}}
\newcommand{\pspace}{(\Omega,\mathcal{F},\mathbb{P})}

\newcommand{\ead}{\mathrm{EAD}}
\newcommand{\lgd}{\mathrm{LGD}}
\newcommand{\npv}{\mathrm{NPV}}
\newcommand{\cf}{\mathrm{CF}}
\newcommand{\pc}{\mathrm{PC}}
\newcommand{\pd}{\mathrm{PD}}
\newcommand{\nd}{\mathrm{ND}}

\newcommand{\elgd}{\mathrm{ELGD} }
\newcommand{\ltv}{\mathrm{LTV}}
 
\newcommand{\rv}{\mathrm{RV}} 
\newcommand{\lc}{\mathrm{LC}}

\newtheorem{theorem}{Theorem}[section]

\theoremstyle{definition}

\newtheorem{algorithm}[theorem]{Algorithm}

\begin{document}

\title{A dimension reduction approach for loss valuation in credit risk modelling}

\def\correspondingauthor{\footnote{Corresponding author: j.he2@uva.nl}}

\author[1]{Jian He \orcidlink{0000-0002-5372-0541} \correspondingauthor{}}
\affil[1]{Korteweg-de Vries Institute, University of Amsterdam, Amsterdam, The Netherlands}
\author[1]{Asma Khedher \orcidlink{0000-0003-1528-3923}}
\author[1, 2]{Peter Spreij \orcidlink{0000-0002-6416-6320}}
\affil[2]{IMAPP, Radboud University, Nijmegen, The Netherlands}

\date{\today} 

\maketitle

\begin{abstract}
\noindent 
This paper addresses the ``curse of dimensionality'' in the loss valuation of credit risk models. A dimension reduction methodology based on the Bayesian filter and smoother is proposed. This methodology is designed to achieve a fast and accurate loss valuation algorithm in credit risk modelling, but it can also be extended to valuation models of other risk types. The proposed methodology is generic, robust and can easily be implemented. Moreover, the accuracy of the proposed methodology in the estimation of expected loss and value-at-risk is illustrated by numerical experiments. The results suggest that, compared to the currently most used PCA approach, the proposed methodology provides more accurate estimation of expected loss and value-at-risk of a loss distribution.                          \smallskip\\
{\sl keywords: Bayesian filter, credit risk, loss valuation} \\
{\sl 2020 Mathematics Subject Classification:}  62P05, 91G40
\end{abstract}
\newpage
\tableofcontents
\newpage

\section{Introduction}

\subsection{Problem description and background}

In this paper we apply {\it the grids methodology} to construct an approximation of {\it the valuation function} which is used to evaluate potential credit losses of loans. To tackle the ``curse of dimensionality'' problem that the grids methodology heavily suffers from,
we propose a dimension reduction technique based on the Bayesian filter and smoother.

Computing risk metrics in credit risk for large portfolios is computationally intensive. This is especially true when the analytical loss distribution is not available. In general, Monte Carlo simulation is required to estimate the loss distribution. In such a case, one needs to simulate the multi-year default probabilities (PD) and realize loss-given default (LGD) of each single loan or each pool of loans, and then evaluate the potential losses. The simulation of the multi-year PD and realized LGD for large and complex portfolios requires extremely large computational resources. Therefore, in practice a valuation grid is usually used to approximate the portfolio valuation functions to avoid multi-year simulations. Valuation grids are a commonly applied computational tool in science and engineering (for example, \citet{thompson1998handbook}) and have also been applied in finance. For instance, \citet{broadie2004stochastic} uses grids to speed up the valuation of high dimensional American options. Recently, the grids method are also used together with the neural networks or Gaussian process regression to obtain better precision and faster calculations in valuations, see \citet{horvath2021deep}, \citet{pagnottoni2019neural}, \citet{Evans2018ScalableGP}. The idea of using grids to measure portfolio risk is used in the stratified sampling methodology by \citet{jamshidian1996scenario} to approximate directly the probability density function of a portfolio value. 
\\
The grids methodology that applies to approximate the valuation function needs only a small number of exact valuations, which are used to construct a grid. The constructed grid is a fast approximation of the valuation function that can then be used during a simulation. The key to applying this method is to be able to express the valuation function as a function of a small number of risk factors, usually not larger than three, since the grids methodology heavily suffers from the ``curse of dimensionality'', \citet{Bellman1957}. The ``curse of dimensionality'' describes the extraordinarily rapid growth of computation time as the number  of the dimension increases. When the dimensionality of the risk factors is large, a large number of valuations is required to construct the interpolation pricing function and hence the grid methodology becomes impractical. Nevertheless, in credit risk modelling, a higher factor model is preferred to better describe the transition probabilities. Therefore, a dimension reduction approach is needed to map the higher dimensional factors to lower dimensional factors that can be used in the valuation grids. Currently, the most used dimension reduction approach is  Principal Component Analysis (PCA), \citet{jolliffe1990principal}. It transforms the original risk factors into a new risk factor space composed of orthogonal principal directions. For example, five dimensional risk factors result in five uncorrelated principal components. Movements of the five dimensional factors can thus be described by the principal components. The principal components are ordered from the largest to the least in terms of explaining the total variance of the five dimensional risk factors. In practice, one needs to decide on the number of principal components to use. Assume one decides to use the first two principal components. Then the original five dimensional risk factors are "reduced" to two dimensional risk factors, which are a linear combination of the first two principal components. However, using the PCA approach in credit risk modelling has the following drawbacks.
\begin{itemize}
  \item \emph{Low interpretability}: principal components are linear combinations of the risk factors, but they are not easy to interpret. For example, usually the risk factors need to be standardized before PCA decomposition. Consequently, it is difficult to tell which are the most important risk factors after computing principal components. Therefore one loses the economic interpretation when analyzing the impact of the principal components.  
  
  \item \emph{Sensitivity to outliers}: PCA is a method based on the correlation or covariance matrix, which can be very sensitive to outliers. Consequently, outliers can completely mess up a classical PCA analysis and yield a PCA model with an arbitrarily bad fit to the genuine part of the data. Indeed, PCA cannot account for outliers in the risk factors that are outside the range of the selected components and such outliers usually need to be removed before performing PCA. However, in credit risk modeling, these outliers are actually desired since they might be located in the tail of the loss distribution and hence contribute to the Value-at-Risk (VaR).   
  \item \emph{Assumption on the correlations}: PCA assumes correlations between the risk factors. If the risk factors are not correlated, PCA will not provide any additional insight.
  \end{itemize}  
As an alternative to the PCA approach, in this paper we propose a dimension reduction methodology. 

\subsection{Contribution}

Our first and main contribution is a dimension reduction methodology based on the Bayesian filter and smoother. The Bayesian filter and smoother estimate the distribution (or value) of a latent process, given the observed data. We interpret the Bayesian filter and smoother as a projection from the observed data to unobserved latent factors. More precisely, this feature is applied to project the \emph{higher} dimensional risk factors onto \emph{lower} dimensional factors.

Compared to the commonly used Principal Components Analysis (PCA) approach, the proposed dimension reduction methodology has the follow advantages.
\begin{itemize}
\item \emph{Economic interpretability}: The proposed approach projects the higher dimensional factors to lower dimensional factors in such a way that the lower dimensional factors capture the most importance features in the model maintaing an economic interpretation. For instance, if one applies the proposed approach to a transition model in which a certain rating migration constitutes the majority of observations (i.e.\ the majority of clients), this rating migration is an importance risk driver and the lower factors will tend to capture the feature of this rating migration. 
\item \emph{Optimality}: The proposed approach is optimal in the sense that the Bayesian filter is optimal on the criterion of Bayes risk of minimum mean square error (MMSE), see \citet{chen2003bayesian}. Intuitively, the proposed approach is optimal in that it seeks the posterior distribution of the lower dimensional factors which integrates and uses all of available information.
\item \emph{It is generic and robust}: This approach is generic in that it applies to a wide range of risk or pricing models,  and insensitive to outliers. 
\end{itemize}

Our second contribution is an application of our proposed dimension reduction approach to credit loss estimation. We review the important concepts and models used in the credit risk modelling, like  transitions, loss given default (LGD), exposure at default (EAD) and loss valuation models. We also derive a closed form formula of Black-Scholes type for the expected LGD of collateral LGD models. We conduct numerical experiments to illustrate the performances of the proposed methodology on the credit loss estimation and we assess the accuracy of the estimation on PD curves, expected losses, and the VaR at different levels. We compare the performance of the valuation grid based on the proposed methodology and the PCA approach. The results indicate that the proposed approach convincingly outperforms the PCA approach.

\subsection{Organization of the paper}

Section~\ref{section:set-up} provides the background knowledge for loss calculations in credit risk modelling. We present the valuation model and the valuation grid approach for the loss calculation. This section also contains brief reviews of transition, LGD and EAD models and specifies the models used in this paper. The proposed dimension reduction methodology is presented in Section~\ref{section:projectionmethod}, in which the Bayesian filter and smoother are also presented. In Section~\ref{section:numerics} we provide numerical experiments to show the performance of the proposed approach. Finally, Section~\ref{section:conclusions} is devoted to the conclusions. In the Appendix we provide a Black-Scholes formula for the expected LGD.

Here follows some notations and other conventions used in this paper.  All random processes are defined on a fixed probability space $(\Omega, \mathbb{F}, \mathbb{P})$. Time is assumed to be discrete. We will use a filtration $\{\mathcal{F}_k, k=0,\ldots, n\}$, where $\mathcal{F}_k$ summarizes the information up to time $k$ and $n$ is the  final time.
We denote by $\cdot^\top$ the transpose of a vector or a matrix.

\section{Set up and background}\label{section:set-up}

In this section we outline the set up, we pose the problem, give a brief survey of models used for loss calculations in credit risk modeling, Bayesian filters (including  Kalman, and particle filter) and address the ``curse of dimensionality'' problem for loss calculations using Bayesian filters.

\subsection{Overview of loss calculations in credit risk}\label{subsection:valuationgrid}

One of the main tasks in credit risk is to forecast and quantify future losses of a portfolio at a certain time horizon. The time horizon is a future time period over which the credit risk is assessed. For instance, in the regulations prescribed by the Basel Committee on Banking Supervision~\citet{CR_BASEL} there is a one-year time horizon across all asset classes. The forecast risk can be quantified, for example, by expected losses or value at risk, depending on the purpose. In the International Financial Reporting Standards (IFRS) extensive reports of the expected losses for different obligors and sectors are required and for Economic Capital (EC) the extreme quantile of the loss distribution (usually the $99.95\%$ quantile) needs to be calculated. 

Three risk parameters are essential in the process of calculating the potential losses: the probability of default (PD), the loss given default (LGD) and the exposure at default (EAD). The PD is the probability that a borrower will fail to pay back a debt. The LGD is the fraction of a loan that is lost when a borrower defaults and the EAD is the predicted amount of the loss, the borrower may be exposed to, when a debtor defaults on a loan. They are generally estimated using the available historical information and are assigned to operations and customers according to their particular characteristics. In this context, the credit rating tools (ratings and scorings) assess the risk for each transaction or customer according to their credit quality by assigning them a score. This score is then used in assigning risk metrics, together with additional information such as transaction seasoning, loan to value ratio, customer segment, etc. The increase in the number of default events in the current economic situation contributes to reinforce the soundness of the risk parameters by adjusting their estimates and by refining methodologies. The incorporation of data from the years of economic slowdown is particularly important for refining the analysis of the cyclical behavior of credit risk. The effect on the PD estimates and on the credit conversion factor is immediate.\footnote{The credit conversion factor calculates the amount of a free credit line and other off-balance-sheet transactions (with the exception of derivatives) to an EAD amount and is an integral part in the European banking regulation since the Basel II accords, see \citet{BASELII}.} An analysis of the impact on the LGD, however, requires waiting for the outcome of the recovery processes associated with the default events. Below we present a brief review to the models of the three risk parameters, i.e.\ PD, LGD and EAD. 

\subsubsection*{Transition models}

A model for the PD, or for transition matrices that govern rating migrations, depends on the rating system. A credit rating system uses a limited number of rating grades to rank borrowers according to their default probabilities. Ratings are assigned by rating agencies such as Fitch, Moody's and Standard \& Poor's, but also by financial institutions themselves. Rating assignments can be based on a qualitative process or on default probabilities estimated with a scoring model. To translate default probability estimates into ratings, one defines a set of rating grades depending on the default probabilities. For example, borrowers are assigned grade AAA if their probability of default is lower than $0.02\%$, to grade AA if their probability of default is between $0.02\%$ and $0.05\%$ and so on. Once the ratings are defined, the concern is to determine the probability with which the credit risk rating of a borrower decreases or increases by a given degree and a transition (PD) model is required. These probabilities that a credit risk rating of a borrower decreases or increases, from one period to the next one, are usually collected in a transition matrix. Namely, the transition matrix measures the probabilities of the migrations between different credit ratings over specific time intervals. In the transition matrix, the element located at the $i$-th row and $j$-th column represents the probability of a borrower migrates from the $i$-th rating to the $j$-th rating. These are called transition probabilities and are often related to macroeconomic variables such as interest rates, inflation, gross domestic product (GDP), unemployment, etc. Alternatively, transition probabilities can also be modelled by using certain abstract latent processes.
\medskip\\
The transition (PD) models in general can be divided into two main classes: structural and reduced form models. 

{\it Structural models} link the up- or down-grade (or default) probabilities of a firm to the value of its assets and liabilities.  The origin of structural credit risk models is the Merton model, \citet{merton1974pricing}, and then further developed by, for instance, \citet{hull2001valuing}, \citet{avellaneda2001distance}, \citet{duffie2004credit}, \citet{albanese2003credit}, and \citet{jeanblanc1999models}. In these structural models, the latent credit process is called a credit index $(X_t)_{t\geq 0}$, reflecting unobservable credit quality over time which is driven by firm-specific variables such as the asset values. A default occurs if the credit index crosses a default barrier. The dynamics of the credit index are specified by the model and the default barriers can be calibrated to historical migration or market data. Structural models can be generalized to describe not only defaults, but also rating migrations. This requires a mapping from the  latent credit index to the rating states at each time, which corresponds to specifying an interval for the credit index that corresponds to each rating class. The intervals are defined by rating barriers, and rating changes occur when the credit index passes a barrier value. 

{\it Reduced form models}, see for example, \citet{jarrow1995pricing}, \citet{duffie2004credit}, \citet{frey2002dependence}, \citet{wendin2006dependent},  \citet{koopman2008multi}, \citet{jeanblanc2020characteristics}, and \citet{jeanblanc2008reduced}, give a description for the latent credit quality of an entity by assuming that it is driven by exogenous factors. They are based on a linear factor model for the latent variable, or equivalently, for the hazard rate of a firm. The key assumption of reduced form models is that the latent process or hazard rate process is driven by systematic external factors. For instance, in a reduced form model the latent credit process $(X_t)_{t\geq 0}$ can be decomposed into two independent parts: a systematic risk factor and a idiosyncratic risk factor. The systematic risk factor includes factors that are relevant for the business environment of the firms, such as the state of the economy, dummy variables for the industry sectors and regions, or other variables that are related to credit quality. The idiosyncratic risk factor is the firm-specific term. As in the structural model, in reduced form model a default occurs if the latent variable $(X_t)_{t \geq 0}$ falls below a certain threshold at a certain time. Same concept can be generalized to credit ratings by defining a mapping from the credit variable $(X_t)_{t\geq 0}$ to the rating states, by using threshold parameters for all ratings. These threshold parameters give the limits for the latent variable to move from current rating to other ratings, see for instance \citet{nickell2000stability}, \citet{frey2002dependence}, \citet{wei2003multi}, \citet{duffie2004credit}. Given the systematic factors, the default (transition) probabilities depends on the  cumulative distribution function of the idiosyncratic factors. This distribution function is usually referred to as the response function. The most common choices for the response function are the probit function (standard normal distribution function) and the logit function. Compared to the structural model, the reduced form models are not explicitly based on on the company's balance sheet. It only requires information generally available in financial markets and hence historical estimation methods can be easily used. But the consequence, and drawback, is that the models do not explain the economic reasons for default. 

In this paper we use a reduced form approach to model the transition probability. We will describe this in more detail in Section \ref{section:tm}.

\subsubsection*{LGD models}

Loss Given Default (LGD) is one of the key determinants of the credit risks of loans and other credit exposures, which measures the severity of the loss when a debtor defaults. It is the share of an asset that is lost when a borrower defaults. An client may not end up with a loss if it is in default since it could be cured. Hence it is always worth noting that the LGD estimate is sensitive to the definition of default,  which can substantially change the implied Recovery rate. Moreover, the LGD is client-specific because such losses are generally understood to be influenced by key transaction characteristics such as the presence or the value of a collateral. For example, in mortgages the houses themselves are usually the collateral. That means that if the borrower defaults, the bank could sell the house to compensate for the loss from missing payments from the borrower. Obviously, the LGD could be stochastic over time. Despite its stochastic nature, statistical modeling of the LGD has been challenging in the academic literature and in banking practice. Therefore many credit risk models used in practice assume that LGD is a deterministic proportion of the exposures subject to impairment and ignore the fact that LGD can fluctuate according to the economic cycle. For example, \citet{altman2011default} and \citet{altman2005link} show that default rates and recovery rates are strongly negatively correlated and measure a correlation of 0.75 between yearly average default rates and loss rates in the United States. They provide strong correlation evidence between macroeconomic growth variables (such as gross domestic product) and recovery rates and test the impact of correlated defaults and LGD.
There are various models proposed to describe the LGD. \citet{frye2000collateral} and \citet{frye2000depressing} propose a structural model with a systematic risk factor representing the state of the economy that drives both defaults and LGD. The dependence of the default indicator and LGD on the common risk factor gives rise to a strong correlation between the two, which is in line with empirical evidence. \citet{hillebrand2006modelling} introduces dependent LGD modeling into a multifactor latent variable framework, providing a good fit to corporate bond data. Marginal distributions for indicator functions and LGD can be specified. \citet{frontczak2015modeling} show that for the retail sector it is important to include collateral and suggested improvements for modelling LGD for this sector. We adopt this latter approach to model the LGD in this paper. The details are presented in Section \ref{sec:LGD-model}.

\subsubsection*{EAD models}

Exposure at Default (EAD) is the predicted amount of loss that may be exposed to when a debtor defaults on a loan. In many cases, such as residential mortgages and personal loans, the EAD are deterministic and can be simply taken from the current on-balance amount. For credit cards though, the revolving nature of the credit line poses challenges with regards to predicting the exposure at default time. As credit card customers may borrow more money (off balance) in the months prior to default, simply taking the current balance for non-defaulted customers would not produce a conservative enough estimate for the amount drawn by the time of default. Examples to model the EAD with an off balance amount can be found in \citet{EAD_CC2013} and \citet{EAD_TONG2016}. 
In this paper, we treat the EAD as deterministic and equal to the on-balance amount. 
Finally, given the risk parameters we next explain how the portfolio losses are determined.

\subsubsection*{Portfolio losses}

Given the PD, LGD and the EAD models, the realized losses of the obligors can be determined by simulating the risk factors and the migrations. If the simulated rating of an obligor becomes default at a certain time $t$, the loss of this obligor can be simply computed as $\lgd_t\times \ead_t$, the notation should be obvious. After discounting and aggregating all the losses at different times, one can obtain the aggregated loss of the obligor. Then the portfolio loss for each realization is computed by summing the aggregated losses from all obligors. 

By simulating a lot of scenarios, one obtains the distribution of the portfolio losses and hence determines the desired risk metrics, such as expected loss or value-at-risk. Figure~\ref{fig:illustration} is presented as a simple illustration to determine the loss distribution of the portfolio. A number of scenarios are simulated until maturity. In this case the maturity is 30 periods and the horizon is one period. In each scenario, the migrations of all obligors are simulated over time and the losses at each time point are computed for the obligors that move to default at that time point. Aggregating all the discounted losses from all time points for all obligors, one obtains the portfolio losses for this specific scenario. In the end the loss distribution can be obtained by repeating the calculations for all simulated scenarios.  
\begin{figure}[hbtp]
 \centering
 \includegraphics[scale=0.5]{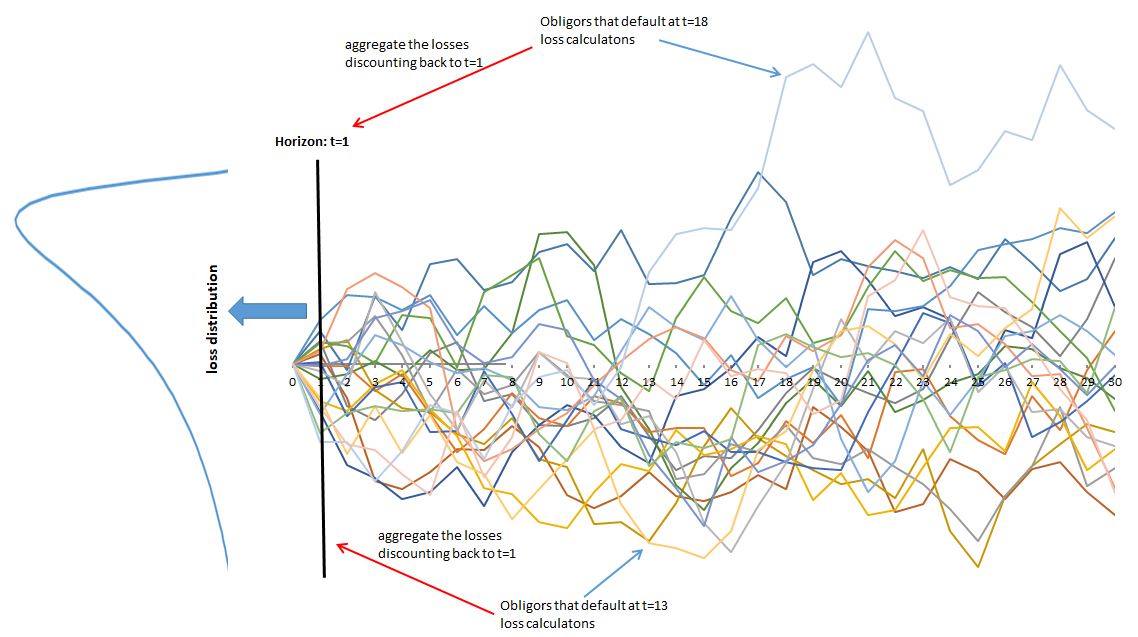}
 \caption{An illustration to determine the loss distribution}
 \label{fig:illustration}      
\end{figure}
This procedure is very time consuming when dealing with a huge portfolio, since it requires simulations of the migrations for all the obligors along each scenario. Hence, in practice, a valuation model together with a valuation grid method is needed to simplify the loss calculations. In the next subsections we introduce the valuation model we use in this paper and we introduce our valuation grid approach to compute the risk metrics.

\subsection{Valuation model for the credit losses}\label{subsection:valuationmodel}

In general, a portfolio's credit loss is defined as the difference between the
portfolio's current value and its future value at the end of some time horizon see \citet{CR_BASEL}. Therefore a valuation model can be used to compute the value of the loans and hence to capture the risks. For example, such a model is used to capture the risk of a one-year loan and a ten-year loan. If, over the one-year horizon, the one-year loan is downgraded but does not default, then it survives and there is no loss. If the ten-year loan is downgraded but does not default, then the \emph{economic} value of this loan should be lower than if this loan was not downgraded. The loan values are based on a net present value (NPV). Specifically the NPV is computed as the expected discounted value of all future net cash flows:
\begin{equation}\label{eq:npv}
\begin{aligned}
\npv = \mathbb{E}\big( \sum_{k=1}^n \delta_k \cf_k\big),
\end{aligned}
\end{equation}
where $n$ is the maturity of the loan, the $\delta_k\in\mathbb{R}^+$ are the risk-free discounting factors and the $\cf_k$ are the future cash flows at time $k$. It is very important to note, while it is correct to use risk-free rates to discount risk-free cash flows, we have to discount risky cash flows due to  potential defaults. This issue can be resolved in two ways: find risk-adjusted discount rates or transform the probabilities to risk adjusted (risk neutral) probabilities. The former way requires finding an appropriate discount rate for each (type of) obligor or loan, which would introduce a lot of parameters in the model. Therefore the latter approach to transform the physical default probabilities into risk-neutral default probabilities is usually used. More details about the risk neutral default probability can be found in \citet{Qmeasure_DG1998}. To reduce the complexity of the model in practice, the risk-free discounting factors are usually derived from deterministic curves, for instance those obtained from the Libor rates or Overnight Indexed Swaps (OIS). In this paper, we do not distinguish between the physical measure and the risk neutral measure for simplicity, and hence the expectation in equation~\eqref{eq:npv} is interpreted as a risk neutral expectation. 
\medskip\\
Consider a loan with coupon payments $s_k$ at time $k=1, 2, \ldots, n$ and a principal repayment $\pc$ at maturity $n$. Since the default depends on the rating of the loan, we denote the default time of a loan with rating $r$ by $\tau_r$, a non-negative real random variable. If the loan defaults in period $k$ (between time ${k-1}$ and $k$), denoted by the indicator $\mathbbm{1}_{\tau_r\in(k-1, k]}$, the cash flow equals the coupon until time $k-1$ plus the recovery $1-\lgd_k$. If the loan does not default, the cash flow equals all the coupons plus the principal. Hence, one obtains the value of the loan with rating $r$ at initial time $k=0$, denoted by $V_0^{(r)}$ as 
\begin{equation*}
\begin{aligned}
V_0^{(r)} &= \mathbb{E}\left[\sum_{k=1}^n \mathbbm{1}_{\tau_r\in(k-1, k]}\big( \sum_{j=1}^{k-1} \delta_{j}s_j + \delta_{k}(1-\lgd_k)\ead_k\big) + \mathbbm{1}_{\tau_r>n}\big( \sum_{k=1}^n \delta_{k}s_k + \delta_{n}\pc\big)\right]\\
&=\underbrace{\mathbb{E}[\sum_{k=1}^n \delta_{k}\mathbbm{1}_{\tau_r\in(k-1, k]}}_{\mathrm{lifetime}\, \pd}(1-\lgd_k)\ead_k]  + \underbrace{\mathbb{E}\left[\sum_{k=1}^n \delta_{k}\mathbbm{1}_{\tau_r>k}s_k \right]}_{\mathrm{coupon\, payments}} + \underbrace{\mathbb{E}\left[\delta_{n}\pc\mathbbm{1}_{\tau_r>n}\right]}_{\mathrm{principle}}.
\end{aligned}
\end{equation*}
Hence, standing at horizon $h$ and suppose the sigma-algebra $\mathcal{F}_h$ contains the information up to time $h$,\footnote{In practice, the specific definition of $\mathcal{F}_h$ depends on the models used for PD, LGD, and EAD. For instance, if factor models are used, the sigma-algebra $\mathcal{F}_h$ is usually defined as the sigma-algebra that is generated by the factors upto time $h$.} the resulting value of a performing loan with rating $r$, denoted by $V_h^{(r)}$, is
\begin{equation}
\begin{aligned}
V_h^{(r)} &= \rv_h + \sum_{k=h+1}^n \delta_{k}\mathbb{E}\left[\mathbbm{1}_{\tau_r\in(k-1, k]}(1-\lgd_k)\ead_k \mid \mathcal{F}_h \right] \\
& + \sum_{k=h+1}^n \delta_{k}s_k\mathbb{E}\left[\mathbbm{1}_{\tau_r>k} \mid \mathcal{F}_h \right] + \delta_{n}\pc\,\mathbb{E}\left[\mathbbm{1}_{\tau_r>n}  \mid \mathcal{F}_h  \right],
\end{aligned}
\label{eq:Vr}
\end{equation}
where $\rv_h$ is the summation of the cash flow until horizon, i.e.\
\begin{equation*}
\rv_h = \sum_{i=1}^h\delta_{i}s_i.
\end{equation*}
Note that $V_h^{(r)}$ is $\mathcal{F}_h$ measurable.

If there is no default after horizon $h$, the value of the loan is just the summation of all the discounted cash flows,
\begin{equation*}
V_{h,\nd}^{(r)} = \rv_h + \sum_{k=h+1}^n \delta_{k}s_k + \delta_{n}\pc.
\end{equation*}
Therefore, one obtains the valuation model for the loss of a performing loan with rating $r$ at horizon $h$, conditioned on information at $h$, as follows.
\begin{equation}
\begin{aligned}
L_h^{(r)} 
& = V_{h,\nd}^{(r)}-V_h^{(r)}\\
& = \sum_{k=h+1}^n \delta_{k}s_k\mathbb{E}\left[\mathbbm{1}_{\tau_r\leq k} \mid \mathcal{F}_h \right] + \delta_{n}\pc\,\mathbb{E}\left[\mathbbm{1}_{\tau_r\leq n} \mid \mathcal{F}_h \right] \\ 
& \qquad - \sum_{k=h+1}^n \delta_{k} \ead_k\, \mathbb{E}\left[\mathbbm{1}_{\tau_r\in(k-1, k]}(1-\lgd_k) \mid \mathcal{F}_h  \right].
\end{aligned}
\label{eq:L}
\end{equation}
Moreover, note that in the last equality we consider the $\ead_t$ as the on-balance amount at time $t$, as in most cases,  and hence it is deterministic. One observes that the loss at time $k$ is from the future cash flows after the defaults, but compensated by the recovery $(1-\lgd_k)\ead_k$. It is worth noticing that the definition of loss in this valuation model is in the \emph{economic} sense. That means once a loan defaults, the losses not only come from the outstanding exposures at default, but also from missing interest payments from the future coupons. Note that the loss in Equation~\eqref{eq:L} is path-dependent. The probability of default depends on the value of the risk factors in the transition model and the $\lgd_k$ may also depend on risk factors in the LGD model.

\subsection{Valuation grid for losses valuation}

There is usually no closed-form expression for the conditional distribution of the cumulative PD or the conditional joint distribution of PD and LGD. Accordingly it is not possible to derive analytical formulas for the conditional expectations $\mathbb{E}[\mathbbm{1}_{\tau_r\leq k} \mid \mathcal{F}_h ]$ and $ \mathbb{E}[\mathbbm{1}_{\tau_r\in(k-1, k]}\lgd_k\mid \mathcal{F}_h  ]$ in Equation~\eqref{eq:L}. As alternative, a valuation grid is usually used to approximate the aforementioned conditional expectations. The valuation grid methodology, see for instance, \citet{chishti1999simulation} and \citet{gibson2000improving}, requires only a small number of valuations, which are used to construct a fast approximation of the valuation function (in our case the conditional expectations in \eqref{eq:L}) when doing the simulation. The basic idea behind the valuation grid approach is simple. Consider a conditional expectation in which the value depends on a small number of risk factors (say two or three). Firstly, the conditional expectation values are evaluated, by Monte Carlo for instance, on a small, pre-selected set of points which forms a grid in the risk factor space. Then, during the simulation to construct the loss distribution, the conditional expectation values under different scenarios are quickly calculated by using the previously computed values on the grid points through interpolation or more sophisticated tools such as Gaussian process regression, see  \citet{rasmussen2006gaussian} or neural networks, see \citet{bishop1994neural} and \citet{angelini2008neural}. 

We apply the valuation grid to every conditional expectation in \eqref{eq:L}, that is to  $\mathbb{E}[\mathbbm{1}_{\tau_r\leq k} \mid \mathcal{F}_h ]$ and to $ \mathbb{E}[\mathbbm{1}_{\tau_r\in(k-1, k]}\lgd_k\mid \mathcal{F}_h ]$ for $k=h+1,\dots, n$.
These conditional expectations are also called \emph{bullets} since they can be seen as the loss of a unit bullet loan\footnote{A unit bullet loan is a loan with unit notional where a payment of the entire principal of the loan, possibly the principal and interest, is due at the end of the loan term.}. In the rest of our paper, we adopt this terminology and refer the conditional expectations as bullets. 
The key to applying the valuation grid method for valuation is to be able to express the bullet as a function of a small number of risk factors. However, these risk factor inputs are often high dimensional, since high dimensional factors are usually needed to describe the transition probabilities. The high dimensionality will result in big challenges in the implementation of the valuation grid, since the number of grid points to maintain a `minimum density' grows too fast. This can be understood intuitively with a simple example. Suppose five risk factors are used to determine the value of a loan.  If one wants to create a valuation grid with ten points in each of the five dimensions, then $10^5$ exact valuations would be required to create the grid. To tackle this problem, we propose a dimension reduction approach to project a higher factor model onto a lower factor model for the valuation. The advantage of the proposed projection approach is that the higher factor model is kept to model the transition probabilities, whereas a corresponding lower factor model is used for the valuation. Further details of the proposed projection approach will be presented in Section~\ref{section:projectionmethod}.

\subsection{valuation grid pre-processing and loss simulation}

Suppose factor models are assigned to model the transitions and LGDs. Denote the risk factors at time $k$ by $x_k$. The filtration $\{\mathcal{F}_s, s=0, 1,..., n\}$ is the information set in which $\mathcal{F}_s$ is the sigma-algebra generated by $\{x_0,...,x_s\}$. For every time $k$, every initial rating and possibly every initial loan-to-value (depending on the LGD model), the values of the bullets $\mathbb{E}\left[\mathbbm{1}_{\tau_r\leq k} \mid \mathcal{F}_h \right]$ and $ \mathbb{E}\left[\mathbbm{1}_{\tau_r\in(k-1, k]}\lgd_k\mid \mathcal{F}_h  \right]$ in Equation~\eqref{eq:L} are pre-computed on certain pre-selected points of risk factors. Then given time $k$, initial rating $r$ and possibly initial loan-to-value, the set of bullet values and the risk factor points form a valuation grid. The pre-computation of the bullets are usually done by Monte Carlo simulation. 

Suppose the default indicator $\mathbbm{1}_{\tau_r\in(k-1, k]}$ and $\lgd_k$ are conditionally independent given  $\mathcal{F}_k$. Then one derives 
\begin{equation*}
\begin{aligned}
\mathbb{E}\left[\mathbbm{1}_{\tau_r\in(k-1, k]}\lgd_k\mid \mathcal{F}_0 \right] &= \mathbb{E}\left[\mathbb{E}\left[\mathbbm{1}_{\tau_r\in(k-1, k]}\lgd_k \mid \mathcal{F}_k \right] \mid \mathcal{F}_0 \right] \\
& = \mathbb{E}\left[\mathbb{E}\left[\mathbbm{1}_{\tau_r\in(k-1, k]}\mid \mathcal{F}_k \right]\mathbb{E}\left[\lgd_k | \mathcal{F}_k \right]\mid \mathcal{F}_0 \right] \\
& = \mathbb{E}\left[\mathbb{P}(\tau_r\in(k-1, k] \mid \mathcal{F}_k)\mathbb{E}\left[\lgd_k | \mathcal{F}_k \right] \mid \mathcal{F}_0 \right] \\
& = \mathbb{E}\left[ \left(\pd_r(k\mid \mathcal{F}_k)-\pd(k-1\mid \mathcal{F}_k)\right)\mathbb{E}\left[\lgd_k | \mathcal{F}_k \right] \mid \mathcal{F}_0 \right],
\end{aligned}
\end{equation*}
with $\pd_r(k\mid \mathcal{F}_k )$ the cumulative PD of loan with initial rating $r$, at time $k$ conditioned on the information $\mathcal{F}_k$. This result suggests that, in the case that $\mathbbm{1}_{\tau_r\in(k-1, k]}$ and $\lgd_k$ are conditionally independent given $\mathcal{F}_k$, one only needs to simulate transition (default) probabilities and expected LGD when using Monte Carlo simulations to estimate $\mathbb{E}\left[\mathbbm{1}_{\tau_r\in(k-1, k]}\lgd_k\mid \mathcal{F}_0 \right]$. The additional simulations of the indicator $\mathbbm{1}_{\tau_r\in(k-1, k]}$ or the realized $\lgd_t$ are not necessary. Therefore, the valuation function \eqref{eq:L} at horizon $h$ can be reformulated as 
\begin{equation}
\begin{aligned}
L_h^{(r)} & = \sum_{k=h+1}^n \delta_{k}s_k \mathbb{E}\left[\pd( k\mid \mathcal{F}_k ) \mid \mathcal{F}_h \right] + \delta_{n}\pc\mathbb{E}\left[\pd(n \mid \mathcal{F}_n ) \mid \mathcal{F}_h \right] \\ 
& \qquad- \sum_{k=h+1}^n \delta_{k}\ead_k \mathbb{E}\left[\left(\pd(k\mid \mathcal{F}_k)-\pd(k-1\mid \mathcal{F}_k)\right)(1-\lgd_k) \mid \mathcal{F}_h  \right].
\end{aligned}
\label{eq:L2}
\end{equation}
When constructing the loss distribution, firstly the paths of the risk factors need to be simulated until the horizon. Then based on the simulated risk factors, the rating transitions and the LGDs can be simulated (computed) according to the transition and LGD models. If the loan defaults or expires before the horizon, the losses are already obtained by the simulation. If the loan is still performing (not defaulted or expired) at the horizon, the valuation grids are then used to valuate the bullets and consequently the loss after the horizon according to the valuation model. Aggregating the losses of all the loans, one obtains the distribution of the aggregated loss.   
In Figure~\ref{fig:illustration_grid} a simple illustration is presented that shows how the valuation grid is used to approximate the loss distribution. Scenarios and migrations are simulated only until the time horizon. Then given each simulation, the realised loss is computed for each loan. Aggregating the losses from all the loans, the loss distribution is obtained by collecting the portfolio losses from all the scenarios. 
\begin{figure}[hbtp]
 \centering
 \includegraphics[scale=0.5]{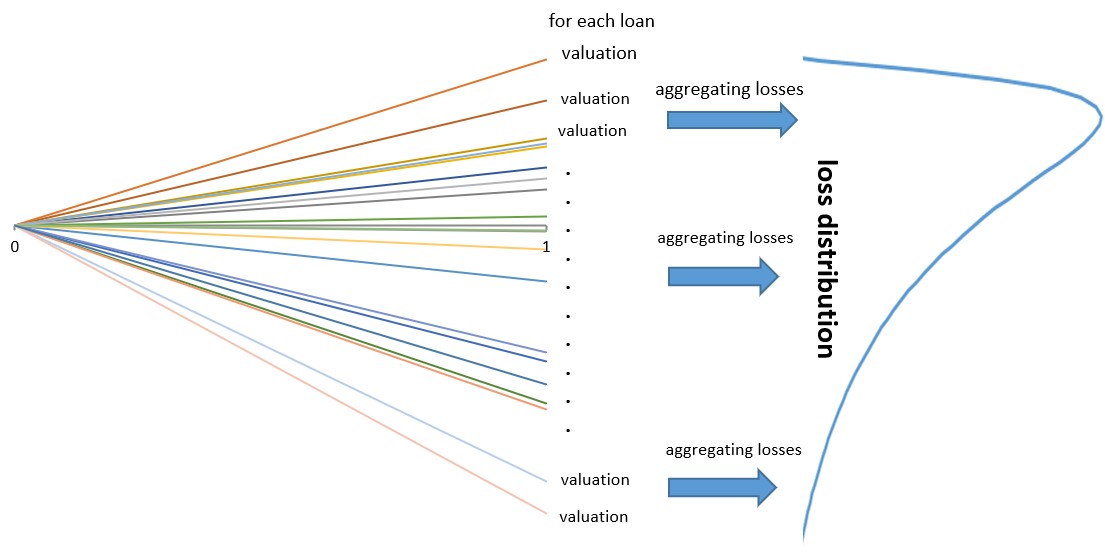}
 \caption{An illustration to determine the loss distribution using valuation grids}
 \label{fig:illustration_grid}      
\end{figure} 

As stated, the losses happened before or on the horizon can be directly computed using the simulated transitions and LGDs and only the losses that happened after the horizon will be evaluated by the valuation grid. Therefore, in this paper we ignore the pre-horizon losses and only focus on the losses after the horizon.


\section{A dimension reduction approach: Bayesian filter projection}\label{section:projectionmethod}

This section on our projection method is the core of the paper and the approach is motivated as follows. 
The desired advantage of the proposed logit transition model, see \eqref{eq:logitT} and \eqref{eq:model-m-theta}, is that it allows a high dimensional multi-factor model which describes the transition data accurately. However, the challenge of implementing a high dimensional transition model in credit risk modeling is the computational complexity in the expected loss valuation using valuation grid. Ideally, the high dimensional transition model is used to simulate the risk factors and determine the transitions, while lower dimensional factors are used for the valuation grid. Note that the simulated higher dimensional factors that are used to determine the transitions and the lower dimensional factors used in the valuation grid need to be consistent. Therefore a (dimension reduction) projection approach is required to map the higher dimensional factors to lower dimensional factors. The most used dimension reduction method approach to tackle this problem is Principal Component Analysis (PCA). But this approach is not optimal to capture the transition probabilities in a lower dimensional space. Hence the performance of this approach is not assured, see the experiments in Sections~\ref{section:experiment_tm}, \ref{section:experiment_EPD} and \ref{section:experiment_el}. In this section, we propose a dimension reduction method, which link the higher dimensional factors with the lower dimensional factor through the Bayesian filter (smoother). Since the Bayesian filter (smoother) achieves minimum mean square error (MMSE) \citet{chen2003bayesian}, the proposed dimension reduction method is optimal for the criterion of the Bayes risk, i.e.\ MMSE, to project the simulated higher dimensional factors to lower dimensional factors. The Bayesian filter is a general probabilistic approach for estimating the (distribution) of the latent states in a state space model. The state space model and the Bayesian filter will be described in Sections~\ref{subsection:state_space_model} and~\ref{subsection:Byesianfilter} respectively. The proposed dimension reduction algorithm is described in Section~\ref{subsection:BFprj}.

\subsection{State space model}\label{subsection:state_space_model}
State-space models deal with dynamic time series problems that involve unobserved variables or parameters that describe the evolution in the state of the underlying system. The general {\it state space} model, defined on some probability space $\pspace$, is as follows
\begin{equation}
\begin{aligned}
x_k &=f_k(x_{k-1},u_k),\,x_0 \\
y_k  &=h_k(x_k,v_k)\,,  \quad k \in \mathbb{N}^+\,,
\end{aligned}
\label{state_space_model}
\end{equation}
where
$f_k:\Real^d \times \Real^p \rightarrow \Real^d$, $h_k: \Real^d \times \Real^q \rightarrow \Real^m$ are given Borel measurable functions, $\{u_k\}_{k\in \mathbb{N}^+}$ are $p$-dimensional
and $\{v_k\}_{k\in \mathbb{N}^+}$ are $q$-dimensional white noise processes both independent of the initial condition $x_0$, and mutually independent as well. Parameters in the functions $f_k$ and $h_k$, together with the covariance of $u_k$ and $v_k$ can be seen as the parameters of the state space model, to which we collectively refer to as $\psi$.

We are interested in estimating the (latent) state process $\{x_k\}_{k \in \mathbb{N}^+}$, but only have access to the process $\{y_k\}_{k \in \mathbb{N}^+}$ which represents the observations. Because of the existence of the white noise in the data, estimating the values, or even the distributions, of the latent states $\{x_k\}_{k \in \mathbb{N}^+}$ by the observations $\{y_k\}_{k \in \mathbb{N}^+}$ is not trivial. There are different methodologies in the literature to estimate the latent process, see e.g.~\citet{Press, ChuiChen, arulampalam2002tutorial}, and the Bayesian filter (smoother) is one of the mostly used approach in practice. Especially, the Bayesian filter (smoother) can be considered as a projection from the observation space to the state space. The Bayesian filter and smoother will be explained in Section~\ref{subsection:Byesianfilter}.

The state space model is widely used in many different areas such as economics and finance, statistics, computer science, and electrical engineering and neuroscience. Particularly in finance, it can be used to identify the latent variables including business cycles, expectations of certain economic variables, permanent income streams, reservation wages, etc. The state space model is also often applied to factor models to describe, for instance, interest rate and credit risk driving factors.

\subsection{Bayesian filter and smoother}\label{subsection:Byesianfilter}

The {\it Bayesian filter} and {\it Bayesian smoother}, see e.g.~\citet{Sarkka2013,Press,Robert} for an overview, are used to estimate the distribution of the latent states $\{x_k\}_{k \in \mathbb{N}^+}$ in \eqref{state_space_model} given the parameters, which we denote  $\psi$. Often, but not always, we emphasize the dependence of  densities on the parameters in the notation.

We let the initial density function of $x_{0}$ be  $p_{0}$. The methodology in Bayesian filtering consists of two parts: prediction and update. At every time point $k$, the prediction part computes (estimates) the {\it prior distribution (density)} of $x_k$ given the past observations up to time $k-1$),  
\begin{equation*}
p(x_k\mid y_{1:k-1},\psi),
\end{equation*}
and the update part computes (estimates) the {\it posterior distribution (density)} of $x_k$ given the past up to time $k$,
\begin{equation*}
p(x_k\mid y_{1:k}, \psi).
\end{equation*}
The purpose of the Bayesian smoother\footnote{This definition actually applies to the fixed-interval type of smoothing. For fixed-leg smoothing, one can consult \citet{anderson2012optimal}.} is to compute the marginal posterior distribution of the state $x_k$ at the time step $k$ after receiving the measurements up to a time step $T$ with $T>k$,
\begin{equation*}
p(x_k\mid y_{1:T}, \psi).
\end{equation*}
The difference between filters and smoothers is that that the filter uses the information from the previous and current period to estimate the posterior distribution of current state $x_k$, while the smoothers uses also the observations after the current period and hence produces optimal estimations.

It follows that the state space model~\eqref{state_space_model} satisfies the properties of a stochastic system, i.e.\ at every (present) time $k\geq 1$ the future states  and future observations $(x_j,y_j)$, $j\geq k$, are conditionally independent from the past states and observations $(x_j,y_{j-1})$, $j\leq k$,  given the present state $x_k$, see~\citet{vanschuppen1989}.
It then follows that  $\{x_k\}_{k \in \mathbb{N}}$ is a Markov process,
and for every $k\geq 1$ one has that $y_k$ and $y_{1:k-1}$ are conditionally independent given $x_{k-1}$, in terms of densities,
\begin{equation}\label{eq:xyk}
p(x_k\mid x_{k-1},y_{1:k-1})=p(x_k\mid x_{k-1}).
\end{equation}
Similarly, due to Markov property, $x_k$ and $y_{k+1:T}$ are independent given $x_{k+1}$, which gives
\begin{equation}\label{eq:xyk+1}
p(x_k\mid x_{k+1},y_{1:T})=p(x_k\mid x_{k+1}, y_{1:k}).
\end{equation}
Moreover, one also has, for every $k\geq 1$,  that $y_k$ and $y_{1:k-1}$ are conditionally independent given $x_{k-1}$, in terms of densities,
 \begin{align}\label{eq:cond-ind-y}
 p(y_k\mid y_{1:k-1},x_k)=p(y_k\mid x_k)\,,\quad  \mbox{for}\quad  k \in \mathbb{N}^+\,.
\end{align}
These three equations, \eqref{eq:xyk}, \eqref{eq:xyk+1} and \eqref{eq:cond-ind-y} are used in the next sections in the derivation of Bayesian filter and smoother recursions.  

\subsubsection{Bayesian filter recursion}\label{subsection:BFR}

Using Bayes' rule and Equation \eqref{eq:xyk}, we deduce that the density function of the prior distribution is given by
\begin{align}
p(x_k\mid y_{1:k-1}, \psi) &= \int p(x_k\mid x_{k-1},y_{1:k-1}, \psi)p(x_{k-1}\mid y_{1:k-1}, \psi)\, \ud x_{k-1}\nonumber\\
                           &= \int p(x_k\mid x_{k-1}, \psi)p(x_{k-1}\mid y_{1:k-1}, \psi)\, \ud x_{k-1}\,.\label{eq:prior}
\end{align}
Note that when $k=1$, the posterior distribution $p(x_{k-1}\mid y_{1:k-1}, \psi)$ is defined as the initial density $p_{0}$. 

The purpose of the Bayesian algorithm is to sequentially compute the posterior distribution  $p_k(\ud x_k\mid \psi)=p(\ud x_k\mid y_{1:k}, \psi)$. Again using Bayes' rule, \eqref{eq:cond-ind-y} and \eqref{eq:prior}, we obtain the posterior 
\begin{align}
p(x_k\mid y_{1:k}, \psi) &= \frac{p(y_k\mid x_k,y_{1:k-1}, \psi)p(x_k\mid y_{1:k-1}, \psi)}{p(y_k\mid y_{1:k-1}, \psi)}\nonumber
\\
& = \frac{p(y_k\mid x_k, \psi)p(x_k\mid y_{1:k-1}, \psi)}{\int p(y_k\mid x_k, \psi)p(x_k\mid y_{1:k-1}, \psi)\, \ud x_k}. \label{eq:posterior}
\end{align}
If we assume the marginal likelihood function $p(y_k\mid x_k, \psi)$ and the transition probability $ p(x_k\mid x_{k-1}, \psi)$ are known, then given the posterior density $p(x_{k-1}\mid y_{1:k-1}, \psi)$ at time $k-1$, we can use Equations \eqref{eq:prior} and \eqref{eq:posterior} to compute the posterior density $p(x_{k}\mid y_{1:k}, \psi)$ at time $k$.
In this way the posterior distributions can be computed recursively given the initial distribution $p_0$.

\subsubsection{Bayesian smoother recursion}\label{subsection:BSR}

According to Bayes's rule and Equations~\eqref{eq:xyk} and \eqref{eq:xyk+1}, the density $p(x_k \mid y_{1:T})$ can be expressed as 
\begin{align}
p(x_k \mid y_{1:T}, \psi)) &= \int p(x_k\mid x_{k+1},y_{1:T}, \psi)p(x_{k+1}\mid y_{1:T}, \psi)\, \ud x_{k+1} \nonumber \\
&= \int p(x_k\mid x_{k+1},y_{1:k}, \psi)p(x_{k+1}\mid y_{1:T}, \psi)\, \ud x_{k+1} \nonumber \\
&= \int \frac{p(x_k, x_{k+1} \mid y_{1:k}, \psi)}{p(x_{k+1} \mid y_{1:k}, \psi)} p(x_{k+1}\mid y_{1:T}, \psi)\, \ud x_{k+1} \nonumber \\
&= \int \frac{p(x_{k+1} \mid x_k, y_{1:k}, \psi)p(x_{k} \mid y_{1:k}, \psi)}{p(x_{k+1} \mid y_{1:k}, \psi)} p(x_{k+1}\mid y_{1:T}, \psi)\, \ud x_{k+1} \nonumber \\
&= p(x_{k} \mid y_{1:k}, \psi) \int \frac{p(x_{k+1} \mid x_k, \psi)}{p(x_{k+1} \mid y_{1:k}, \psi)} p(x_{k+1}\mid y_{1:T}, \psi)\, \ud x_{k+1}. \label{eq:smoother} 
\end{align}
According to Equation~\eqref{eq:smoother}, given that the prior and posterior distribution are known from the Bayesian filter recursions, the Bayesian smoother estimation $p(x_k \mid y_{1:T}, \psi)$ can be computed using a backward recursion.  
\medskip\\
While the mathematical derivations of the Bayesian filter and smoother are straightforward, in practice it is always a big challenge to compute the integrals in Equations~\eqref{eq:prior},  \eqref{eq:posterior} or \eqref{eq:smoother}. Although there are some special cases (Gaussian linear cases) where theoretical formulas are available for the integrals, such as Kalman filter (see \citet{welch1995introduction} for instance),  in most cases one has to find approximations or  numerical algorithms to compute these integrals. For the extended Kalman filter,  see, for example, \citet{einicke1999robust} and \citet{wan2001dual} and for the unscented Kalman filter, see~\citet{wan2001unscented}, the prior and the posterior distributions are approximated by Gaussian distributions. \citet[Section~10.6--10.7]{durbin2012time} proposes a mode estimation approach to compute the maximum-a-posterior (MAP) estimate of the posterior distribution. A more general approach to estimate the prior and posterior distribution is the particle filter, see e.g.~\citet{arulampalam2002tutorial,cappe2007overview,doucet2009tutorial}, in which prior and posterior distributions are estimated by a Monte Carlo method. For the implementation of the particle filter one can also refer to \citet{CJP, CrisanMiguez, he2021kalman}.

\subsection{Bayesian filter/smoother projection}\label{subsection:BFprj}
As described in Section~\ref{subsection:Byesianfilter}, the Bayesian filter and smoother can be seen as a projection from the observations $y_{1:k}$ to the (distribution of) state variables, given the model parameters $\psi$. From this perspective, the Bayesian filter (smoother) can be used as an approach to project higher dimensional factors to lower dimensional factors. In this section we will illustrate the Bayesian filter (smoother) projection approach in the credit risk environment specified in Sections \ref{subsection:valuationmodel}. However, note that this approach can also easily be applied or extended to other models for different risk types, such as market risk or interest rate risk, see also the examples in Section~\ref{section:examples}.
\medskip\\
Suppose transition models, see Section~\ref{subsection:valuationgrid}, with higher and lower dimension are calibrated based on the same data, and denote the calibrated parameters of higher and lower dimensional model by $\psi_H$ and $\psi_L$, respectively. The higher dimensional factor model is used to simulate the transition probabilities until the horizon  according to the transition model, for instance, a simulation based on Equations~\eqref{eq:model-m-theta} and \eqref{eq:logitT} in Section~\ref{section:numerics}. Denote the simulated higher dimensional factors by $x_{k, H}$ and the corresponding transition probabilities by $T(x_{k, H})$, for $k=1,\ldots,h+1$. One would like to use a lower factor model for a loss calculation. It is then required to project the higher dimensional factors $x_{k, H}$ to the lower dimensional factors, denoted by $x_{k, L}$. Note that the simulated transition probabilities $T(x_{k, H})$, based on the higher dimensional factors $x_{k, H}$, can be used as the observations of the lower factor model. Therefore, using the Bayesian filter and smoother iterations, one can obtain a projection from $x_{k, H}$ to $x_{k, L}$,
as described in Algorithm \ref{algorithm:projection}.

\begin{algorithm}[{\bf projection}]\label{algorithm:projection}
Given are the calibrated parameters of the higher and lower factor model $\psi_H$ and $\psi_L$. Suppose  $x_{1:h+1, H}$ and  $T(x_{1:h+1, H})$ are the simulated higher dimensional factors and the corresponding transition probabilities from time $k=1$ to horizon $h+1$.
\begin{enumerate}
\item Specify the initial density $p(x_{0}\mid \psi_L)$ for the lower dimensional factors. 
\item Using the transition probabilities $T(x_{1:h+1, H})$ as the data,
	\begin{enumerate}
		\item \label{algm:BF}  for $k=1,\dots,h+1$, compute the prior and posterior distribution $x_{k, L}$ according to the Bayesian filter iterations, see  Section~\ref{subsection:BFR}, i.e.		
		\begin{equation*}
		\begin{aligned}
		p(x_k\mid T(x_{1:k-1, H}), \psi_L) &= \int p(x_k\mid x_{k-1}, \psi_L)p(x_{k-1}\mid T(x_{1:k-1, H}), \psi_L)\, \ud x_{k-1}\, \\
		p(x_k\mid T(x_{1:k, H}), \psi_L) & = \frac{p(T(x_{k, H})\mid x_k, \psi_L)p(x_k\mid T(x_{1:k-1, H}), \psi_L)}{\int p(T(x_{k, H})\mid x_k, \psi_L)p(x_k\mid T(x_{1:k-1, H}), \psi_L)\, \ud x_k};
		\end{aligned}
		\end{equation*}
		\item \label{algm:BS} given prior and posterior distributions computed in  step \ref{algm:BF}, for $k=h,\ldots,1$, compute the Bayesian smoother using the backward recursion, see  Section~\ref{subsection:BSR},
		\begin{equation*}
		\begin{aligned}
			p(x_k\mid T(x_{1:h+1, H}), \psi_L) &= p(x_{k} \mid T(x_{1:k, H}), \psi_L) \times \\
			& \int \frac{p(x_{k+1} \mid x_k, \psi_L)}{p(x_{k+1} \mid T(x_{1:k, H}), \psi_L)} p(x_{k+1}\mid T(x_{1:h+1, H}), \psi_L)\, \ud x_{k+1};
		\end{aligned}
		\end{equation*}
		\item determine the lower dimensional factor $x_{k, L}$ projected from $x_{k, H}$ as equal to
		\begin{equation*}
		\mathbb{E}\left[p(x_{k}\mid T(x_{1:h+1, H}), \psi_L)\right], \, k=1,\ldots,h.
		\end{equation*}
	\end{enumerate}
\end{enumerate}
\end{algorithm}
\noindent
We are in particular interested in the value of the lower dimensional factor at horizon $h$, i.e.\ $x_{h, L}$, since it is finally used as the input of the valuation grid to evaluate the losses after this horizon. Note that in order to obtain the optimal estimate of $x_{h, L}$, the simulation is generated until $h+1$ so that the Bayesian smoother, step~\ref{algm:BS} in Algorithm~\ref{algorithm:projection}, can be used for $x_{h, L}$. Figure~\ref{fig:projection} presents the diagram of the projection approach. 
\begin{figure}[hbtp]
 \centering
 \includegraphics[scale=0.5]{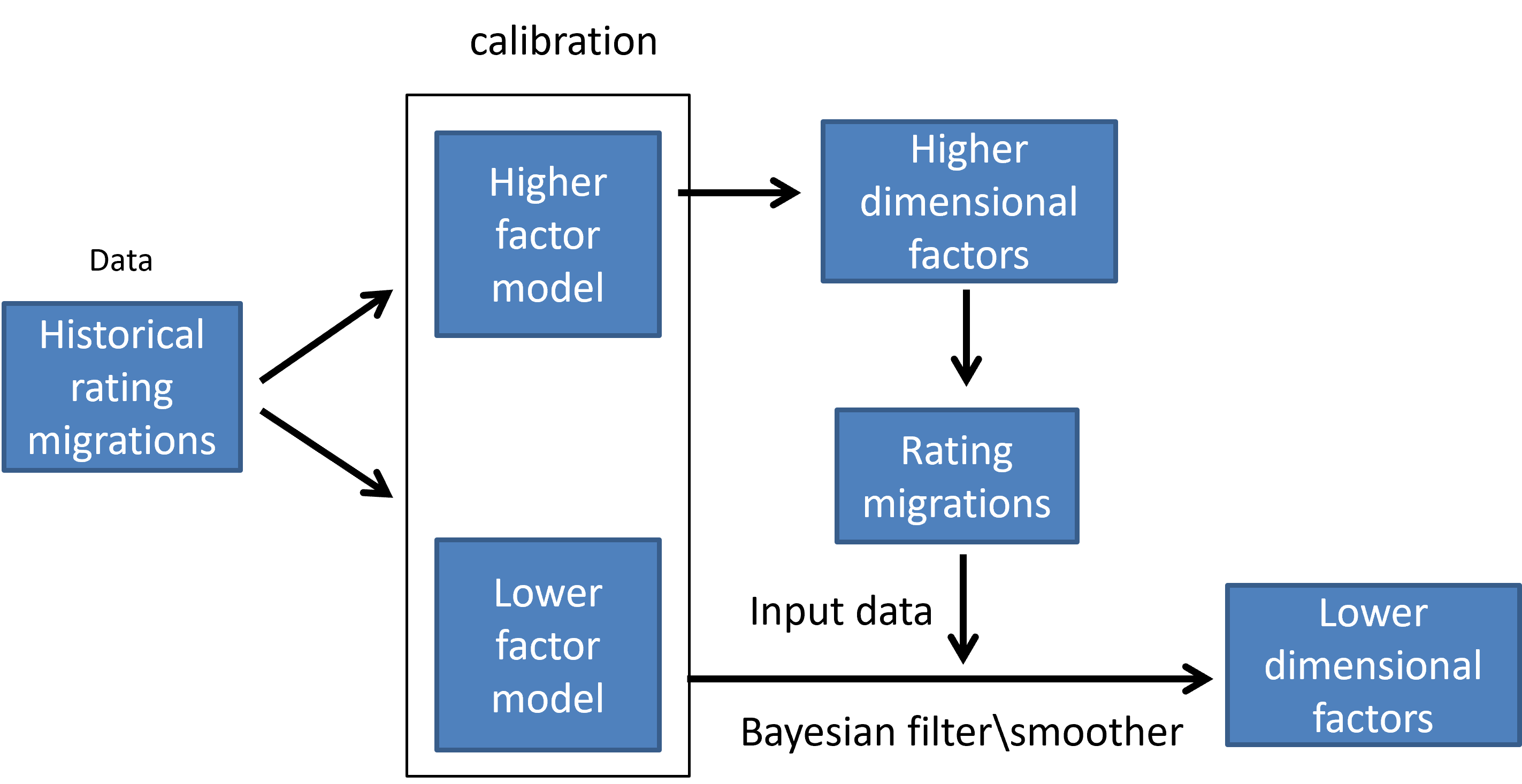}
 \caption{High level description of projection algorithm}
 \label{fig:projection}      
\end{figure}

\begin{remark}
An important input for the iteration of the projection is the initial distribution of $x_{k, L}$, i.e.\ $p(x_{0}\mid \psi_L)$. In credit risk, there are two methodologies to specify the initial distribution (or starting point) of the simulation, the point-in-time (PiT) or through-the-cycle (TtC). In PiT, all the simulations start at the same point and this point is defined as the `state of today', i.e.\ the value of the risk factor at the last period used in the calibration. In TtC, instead of defining a single starting point, an initial distribution is assigned to sample the starting point of each path. This initial distribution is usually defined as the equilibrium distribution of the risk factors process.\footnote{In practice, constraints are always imposed to make sure that the risk factor process is stationary and hence the equilibrium distribution exists.} Therefore, in both PiT and TtC cases, the initial distribution of $x_{k, L}$ can be easily defined.
\end{remark}
Integrating the projection algorithm into the calibration, simulation and valuation algorithms, we have our engine for the risk calculation, illustrated by Figure~\ref{fig:algorithm2}, and described in the following Algorithm~\ref{algorithm:projection2} , formulated in high level terms.
\begin{algorithm}[{\bf calibration, training, simulation, projection and valuation}]\label{algorithm:projection2}~ 
\begin{itemize} 
\item \textbf{Calibration}
	\begin{enumerate}
	\item 
	The transition matrices are modeled by a higher factor model, which is calibrated using historical data. One obtains $\psi_H$, the calibrated parameters of  the higher dimensional model. 
	\item 
	A lower factor model is also calibrated using the historical data and one obtains $\psi_L$, the calibrated parameters of the lower dimensional model.
	\end{enumerate}
\item \textbf{Training of the interpolation grid} 
\medskip\\
Define the higher dimensional grid, perform then for each rating the following steps.
	\begin{enumerate}
	\item 
	Starting at the points on the defined grid, use the higher factor model to simulate the scenarios of the risk factors and determine the expected losses.
	\item 
	Project the points (factors) of the higher dimensional grid to obtain lower dimensional points (factors) using the projection approach as in Algorithm~\ref{algorithm:projection}.
	\item 
	Use the lower dimensional points and the simulated expected losses to build the valuation grid, using interpolation or other techniques like Gaussian process regression, Neural networks, etc.
	\end{enumerate}
\item \textbf{Simulation}
	\begin{enumerate}
	\item 
	Using the high factor model, generate the factors and hence the transition matrices until the horizon.
	\item \label{sim:step2}
	Based on the generated transition matrices, simulate the rating migrations of the obligors until the horizon.
	\end{enumerate}
\item \textbf{Projection and valuation}
	\begin{enumerate}
	\item 
	Project according to the Projection Algorithm~\ref{algorithm:projection} the higher dimensional factors simulated in Step~\ref{sim:step2} of the Simulation part above to lower dimensional factors.
	\item 
	Use the valuation grid to evaluate the expected losses, given the ratings at the horizon and the lower dimensional factors.
	\item 
	The desired risk metrics, such as expected loss and value-at-risk, can finally be obtained by aggregating the expected losses.
	\end{enumerate}
\end{itemize} 
\end{algorithm}

\begin{figure}[hbtp]
 \centering
 \includegraphics[scale=0.5]{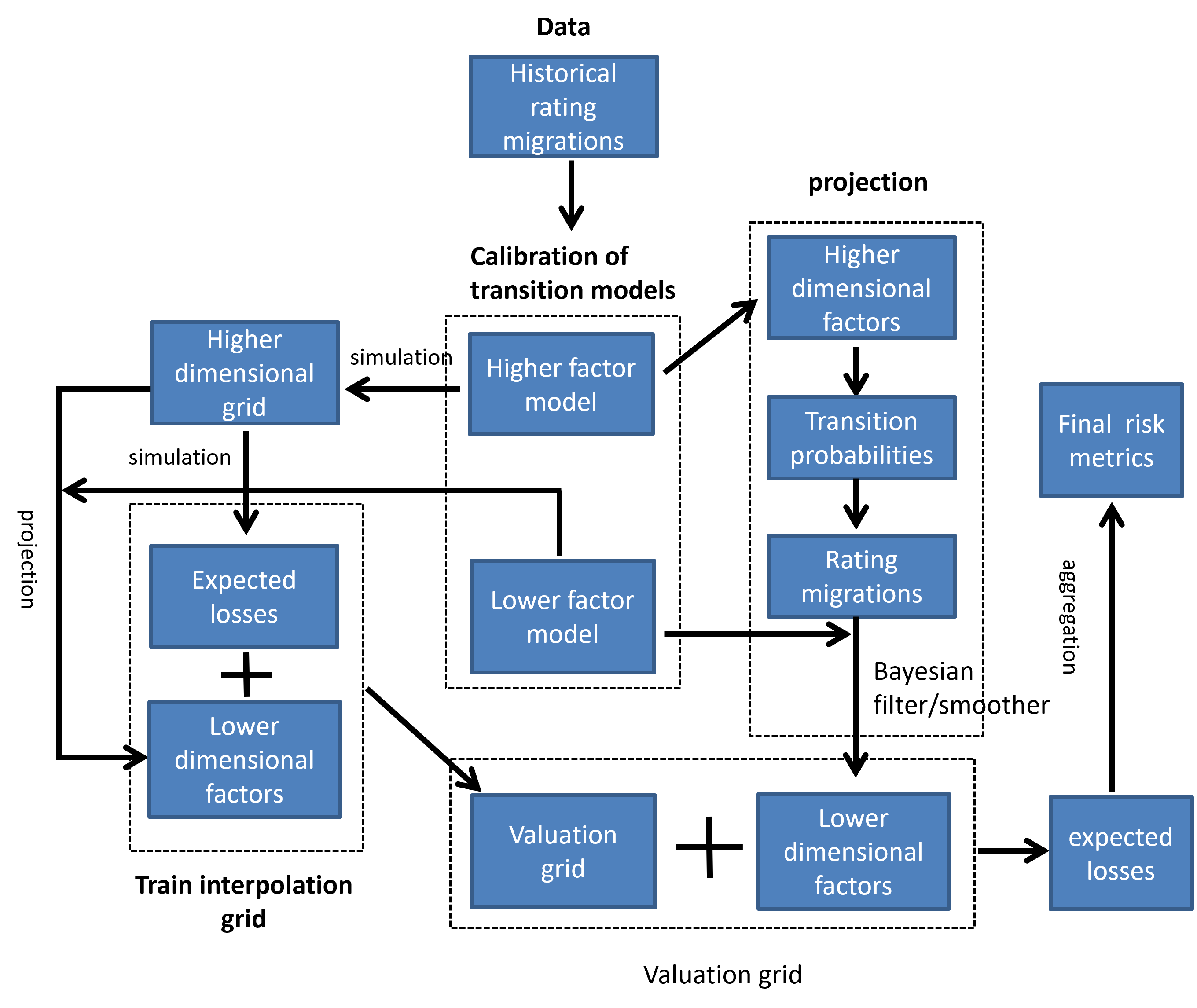}
 \caption{High level description of the algorithm}
 \label{fig:algorithm2}      
\end{figure}

\subsection{Application examples}\label{section:examples} 
The proposed Bayesian projection approach is a generic dimension reduction method which can be used for different types of credit risk valuations.
It is originally designed for the credit risk valuations, but it can also be applied to other grid based valuation models where the dimension reduction approach is needed. We list some possible directions for further  applications of the proposed Bayesian projection approach in the following subsections. 

\subsubsection{Credit Risk Economic Capital}
Economic Capital, mainly for financial service forms, is a metric to determine the amount of capital a firm must hold to protect itself against economic risks (such as market risk, credit risk, and operational risk), given its mixture of assets and liabilities and its strategy to risk appetite. The risks are measured by the unexpected losses through client defaults and the asset value changes through the movements in the market parameters. Typically the losses are estimated based on a time horizon and a confidence level. The time horizon is usually one year and the confidence level is usually $99.9\%$ or $99.95\%$, implying that the Economic Capital should cover the losses that would arise  in case an adverse event would occur with a probability of $0.1\%$ or $0.05\%$ over one year horizon. In particular, the Credit Risk Economic Capital (CREC) is to measure the internal capital required to cover the unexpected credit losses under ECB ICAAP Guidelines, see \citet{ICAAP}.

To measure the CREC, one particularly looks into two statistics: the expected loss (EL) and the $99.9\%$ or $99.95\%$ quantile of the portfolio loss. The CREC covers both default risk and migration risk. The default risk captures the losses from the default event within the time horizon, while the migration risk concerns the potential economic losses associated with non-default credit rating migrations. Although a non-default transition may not result in a direct loss under accrual accounting at the time horizon, it may result in an indirect loss in terms of the expected future cashflows, interest, principal and recoveries. According to this nature, the calculation of the CREC risk measures can be divided into two parts: the realised default losses within the time horizon and the potential losses after the time horizon. A Monte Carlo approach can be directly used to estimate the default losses within the time horizon, while a valuation grid is needed to compute the potential losses after the time horizon. Higher dimensional factors are desired for the simulation of the defaults and rating migrations up to the horizon, due to the higher precision. Meanwhile, a dimensional factors are used as the inputs of the grids to compute the potential losses after horizon. Therefore, the proposed Bayesian dimension reduction approach can be implemented to project higher dimensional factors to lower dimensional factors for the valuation of the potential losses after the horizon. 

\subsubsection{Incremental Risk Charge}
The Incremental Risk Charge (IRC), see \citet{IRCguide},  is a complement to the traditional stressed Value-at-Risk measured in Market Risk Economic Capital.  It is an estimate of default and migration risk of unsecuritized credit products in the trading book. The concept and calculation of IRC is in general similar to CREC. The IRC also measures the default and migration risks at a $99.9\%$ level of confidence over a one-year horizon. However, the CREC model mainly applies to the banking book, including exposures to government, institutions, corporates, retail exposures secured by real estate, and revolving retail exposures. It also covers the (one-year) counterparty credit risk of derivative instruments, securities financing transactions and hold-to-maturity assets in the trading book. By comparison, the IRC mainly concerns the credit risks of the bond and equity exposures in the trading book. Different from the booking book, in which the positions are held to maturity, the positions in the trading book are continuously bought and sold before their maturities. Therefore, unlike CREC, the IRC model differentiates the underlying traded instruments by liquidity horizon, with a minimum of three months\footnote{Note that the capital horizon is still one year.}. A constant level of risk assumption \footnote{The constant level of risk means the risk of the portfolio remains to be its initial risk level over the one-year horizon. In practice, This can be achieved by holding the constant portfolio over the liquidity horizon, rebalancing any default, downgraded, or upgraded positions at the beginning of each liquidity horizon, and rolling over any matured positions at the beginning of each horizon.} is also imposed in IRC and ensures that all positions in the IRC portfolio are evaluated over the full one-year time horizon. Since the calculation of IRC also requires migration/default simulations for every liquidity until a one year capital horizon, and the migration risks need to be evaluated according to certain valuation function, the proposed Bayesian dimension reduction approach can be applied to make the calculations more efficient.    

\subsubsection{International Financial Reporting Standard 9}
The International Financial Reporting Standard 9 (IFRS 9), see for example \citet{IFRS9instrument}, \citet{esrb2019expected}, \citet{gornjak2020literature}, is published by the International Accounting Standards Board (IASB). It addresses the accounting for financial instruments and it replaced the International Accounting Standard (IAS) 39 with regard to the methodology used to compute impairment provisions on financial instructions. IFRS 9 introduces a new impairment model based on the Expected Credit Loss (ECL), which is different from the Incurred Credit Loss (ICL) model under IAS 39. Under IFRS 9, institutions, as reporting entities,  will have to recognize not only incurred credit losses but also losses that are expected in the future. In IFRS 9, two different ECLs are required, depending on the stages of impairment, see \citet{IFRS9ECL}. In stage 1, the ECLs result from the default events that are possible within the next 12 months (12-month ECL). While in stage 2 or 3, lifetime ECLs are recognised. To obtain a 12-month ECL, a Monte Carlo simulation can be used based on the given PD and LGD models. While, for the lifetime ECL, a grid approach is required\footnote{Sometimes the 12-month ECL is used to construct a proxy for the lifetime ECL in the simplified approach, see \citet{IFRS9Monitoring}.}. Since the valuation of life time ECL usually involves a large number of factors from the PD and LGD models, the proposed Bayesian projection approach can be applied to reduce the dimensionality of the valuation grid for ECL and hence improve the efficiency of the computations.  

\subsubsection{Market risk}
Market risk can be defined as the risk of losses in on and off-balance sheet positions arising from adverse movements in market prices. From a regulatory perspective, market risk stems from all the positions included in the trading book as well as from commodity and foreign exchange risk positions in the whole balance sheet. In  market risk calculations, the VaR needs to be determined according to the future Profit-and-Loss (P\&L) distribution of the portfolios. The value of the portfolio depends on the future values of the asset prices, equity market volatilities, interest rates, etc. Therefore, to obtain the P\&L distribution, one has to simulate the future scenarios of (the risk factors of) the relevant asset prices, equity market volatilities or interest rates curves, and re-evaluate the portfolio values given the simulated scenarios. In order to efficiently re-evaluate the portfolio values, a valuation grid is usually used, see for instance \citet{gibson2000improving}, \citet{zamani2022pathwise}. The valuation grid links the risk factors and the derivative values. However, many derivatives involve multiple underlyings, for instance the multi-asset options (\citet{leentvaar2007multi}, \citet{ekedahl2007dimension}). Therefore, the proposed  Bayesian projection approach can be used in pricing such derivatives to avoid the ``curse of dimensionality''. 

\subsubsection{Interest rate risk in the banking book}
Interest rate risk in the banking book (IRRBB) refers to the current or prospective risk of the bank’s capital and earnings arising from adverse movements in interest rates that affect the bank’s banking book positions. When interest rates change, the present value and timing of future cash flows change. This in turn changes the underlying value of a bank’s assets, liabilities and off-balance sheet items and hence its economic value. Changes in interest rates also affect a bank’s earnings by altering interest rate-sensitive income and expenses, affecting its net interest income (NII). The calculation of the VaR for IRRBB is similar to the calculation of VaR for market risk. The P\&L is generated by simulating the (risk factors of) interest rates curves and re-evaluating the balance sheet given the simulated scenarios. The re-evaluation of the balance sheet can be done by using the valuation grid method. In the past years since the credit and liquidity crisis of 2007, a multi-curve framework of interest rates was developed to replace the single-curve framework. The single-curve framework bootstrapped a single risk-free curve which represented both the cost of funding future cash ﬂows and forward rates. This unique curve is used for both discounting and forecasting, while in the multi-curve framework, the discount curves and the forecasting curves are distinguished. As expected, the different forecasting tenor curve forecasts the future cashﬂows of different tenors and the discounting curve values these cashﬂows to today. More details regarding the single- and multi-curve framework can be found in, for instance, \citet{henrard2014interest}, \citet{pallavicini2010interest}. The multi-curve framework introduces more curves in pricing the derivatives. For example, three curves are needed to price the Constant Maturity Swaps or the Basis Swaps, i.e.\ one discounting curve and two forecasting curves. If each curve requires a two-factor model, in total six factors are needed to construct a valuation grid for these derivatives. To decrease the dimensionality of the valuation grid, one can apply the proposed  Bayesian projection approach.

\section{Numerical studies}\label{section:numerics}

In the sequel we illustrate the performance of the projection approach by examples. We aim to forecast the bullets in Equation~\eqref{eq:L2}, and hence forecast the loss distribution of a loan. The chosen loan in the experiment is a unit bullet loan, of which the notional is 1. Consequently the principal repayment $\pc=0$ and $\ead_k =1$ for $k=1,\ldots,n$. The coupon payments $s_k$ are set to be equal to $s = \frac{1}{n}, k=1,\ldots,n$. For convenience and without loss of generality,  discounting is not considered in this experiments. Hence Equation~\eqref{eq:L2}  simplifies to

\begin{equation}
\begin{aligned}
L_1^{(r)} & = \sum_{k=2}^n s \mathbb{E}\left[\pd( k\mid \mathcal{F}_k ) \mid \mathcal{F}_h \right] + \mathbb{E}\left[\pd(n \mid \mathcal{F}_n ) \mid \mathcal{F}_h \right] \\ 
& \qquad - \sum_{k=h+1}^n  \mathbb{E}\left[\left(\pd(k\mid \mathcal{F}_k)-\pd(k-1\mid \mathcal{F}_k)\right)(1-\lgd_k) \mid \mathcal{F}_h  \right].  
\end{aligned}
\label{eq:lsimple}
\end{equation}
In all experiments we assume the maturity of the unit loan after a horizon of 30 periods, i.e.\ $n=30$. 
\medskip\\
The value of the bullets depends on the transition and LGD models, therefore, we next introduce the transition and LGD models used in the experiments.

\subsection{Transition model}\label{section:tm}

The transition model we use belongs to the reduced form model. Let $k=1,\ldots,n$. Suppose we have $R>2$ credit states. We denote by $T_{ij,k}$ the probability at time $k$ of the transition from a credit state $i=1,\ldots, R$ to a credit state $j=1,\ldots,R$. We model the transition probabilities (or observed transition rates) using a logit function as follows,
\begin{equation}\label{eq:logitT}
T_{ij,k} =  \frac{g_{ij}\exp(\theta_{ij,k})}{\sum_{j=1}^R g_{ij}\exp(\theta_{ij,k})}\,,
\qquad i,j =1,\ldots R\,,\quad k=1,\ldots, n\,.
\end{equation}
The $\theta_{ij,k} + \log{g_{ij}}$ are usually referred to as the log-odds for the transition from a credit state \textit{i} to a credit state \textit{j}, in which the $g_{ij}$ are the parameters which indicate the level of the transition probabilities and $\theta_{ij,k}$ are the signals that describe the dynamics of the transition probabilities. Suppose, for clients in the $i$-th rating, that in $N_{i, k}$ independent trials the number of {\it observed} migrations of these clients  to the $j$-th credit state is $m_{ij,k}$ for $j = 1,\ldots, R$. It is easy to see that $m_{i, k}=[m_{i1, k}, \ldots, m_{iR, k}]$ has a multinomial distribution, with log-density given by
\begin{equation}\label{eq:likelihood-m}
\log p(m_{i, k} |T_k)= \sum_{j=1}^{R} m_{ij,k}\log T_{ij,k}+ \log C_{i,k}\,,
\end{equation}
where 
$$C_{i, k} = \frac{N_{i, k}!}{\prod_{j =1}^{R} m_{ij,k}!}.
$$
Plugging Equation~\eqref{eq:logitT} into Equation~\eqref{eq:likelihood-m}, 
one can rewrite the log-likelihood function \eqref{eq:likelihood-m} at time $k$ as
\begin{equation*}
\log p(m_{i, k}|\theta_{ij, k}) = \sum_{j=1}^R m_{ij, k}(\theta_{ij, k} + \log g_{ij})  - N_{i, k} \log\big(\sum_{j=1}^R g_{ij}\exp(\theta_{ij,k})\big) + \log(C_k)\,.
\end{equation*}
\noindent We assume that the migrations from a state $i$ to a state $j$ are driven by the {\it  latent} common factors $x$  (hence the migrations are also correlated). The latent process $x$ is assumed to evolve linearly over time with a Gaussian error $\eta$, and therefore our transition model is given by: 
\begin{equation}\label{eq:model-m-theta}
\begin{aligned}
m_{i, k} &\sim p(m_{i, k}|\theta_{ij, k})\,, 
\quad  \theta_{ij, k}  = K x_k, \quad k=1,\ldots,n\,, \\
x_{k} & = A x_{k-1} + \eta_{k}\,, 
\quad \eta_{k} \sim \mathcal{N}(0,Q),  
\end{aligned}
\end{equation}
where $x_k$ is a $d \times 1$ random vector,  $K$ is a $(R-1) \times d$-matrix, and $A$ and $Q$ are $d \times d$-matrices. The $g = \{g_{ij}, i, j=1, \ldots, R\}$, $K$, $A$ and $Q$ are referred to as the parameters of the transition model. The initial distribution of $x_k$ is assumed to be Gaussian with mean $a_0$ and (co)variance $P_0$, i.e.\ $x_0 \sim \mathcal{N}(a_0,P_0)$.\\ 
\\
This transition model belongs to the class of \emph{state space models} of Section~\ref{subsection:state_space_model}. The latent states $x_t$ and the model parameters are usually estimated by the Bayesian filter introduced in Section~\ref{subsection:Byesianfilter}. Specially, we estimates the latent states $x_t$ by using the method of \emph{mode estimation} (\citet[Section 10.6]{durbin2012time}) based on the Kalman filter.
\medskip\\
We suppose to have four different ratings with three performing ratings (P1, P2 and P3) and one default rating (D). We further assume that the default state is absorbing, which means that the probabilities of the transition from D to P1, P2 or P3 are zero. The set-up of the transition model in this example is the same as in Equations~\eqref{eq:model-m-theta} and \eqref{eq:logitT}. In this experiment we have three models: one benchmark model and two testing models. The benchmark model is chosen to be a four factor model, i.e.\ the dimensionality of $x_k$ is four. This model is treated as the `true' model and used to generate the loss  distribution which serves as a benchmark for comparison. The two testing models are chosen to be a two factor model and an one factor model. These two models are used to assess the performance of the proposed dimension reduction approach. All models are calibrated using the same migration data with a horizon of 120 periods. These migration data are simulated by using the benchmark model with parameters as follows. The parameter values are obtained from further classified real data, which we are not allowed to disclose.
\begin{align*}
A &= \diag([0.6, 0.95, 0.9, 0.5]),\\
Q &= \diag([0.6, 0.1, 0.1, 0.7]),\\
K &= 
\begin{pmatrix}
k_{11,1} & k_{11,2} & k_{11,3} & k_{11,4}\\
k_{12,1} & k_{12,2} & k_{12,3} & k_{12,4} \\
k_{13,1} & k_{13,2} & k_{13,3} & k_{13,4} \\
k_{14,1} & k_{14,2} & k_{14,3} & k_{14,4} \\
k_{21,1} & k_{21,2} & k_{21,3} & k_{21,4} \\
k_{22,1} & k_{22,2} & k_{22,3} & k_{22,4}\\
k_{23,1} & k_{23,2} & k_{23,3} & k_{23,4}\\
k_{24,1} & k_{24,2} & k_{24,3} & k_{24,4}\\
k_{31,1} & k_{31,2} & k_{31,3} & k_{31,4}\\
k_{32,1} & k_{32,2} & k_{32,3} & k_{32,4}\\
k_{33,1} & k_{33,2} & k_{33,3} & k_{33,4}\\
k_{34,1} & k_{34,2} & k_{34,3} & k_{34,4}\\
k_{41,1} & k_{41,2} & k_{41,3} & k_{41,4}\\
k_{42,1} & k_{42,2} & k_{42,3} & k_{42,4}\\
k_{43,1} & k_{43,2} & k_{43,3} & k_{43,4}\\
k_{44,1} & k_{44,2} & k_{44,3} & k_{44,4}\\
\end{pmatrix}
= 
\begin{pmatrix}
0 	  &  0    & 0 	  & 0\\
0.2   & 0.04  & 0.06  & 0.1\\
0.12  & -0.36 & 0.12  & -0.04\\
0.38  & -0.28 & -0.26 & -0.08 \\
-0.17 & 0.34  & 0.18  & -0.01 \\
0 	  &  0    & 0 	  & 0 \\
0.02  & -0.27 & 0.08  & -0.01 \\
-0.07 & -0.14 & 0.16  & -0.05\\
-0.22 & -0.11 & 0.01  & 0.12 \\
-0.03 & -0.2  & 0.01  & 0.22 \\
0 	  &  0    & 0 	  & 0 \\
-0.08 & 0.09  & -0.05 & -0.03 \\
0 	  &  0    & 0 	  & 0\\
0 	  &  0    & 0 	  & 0\\
0 	  &  0    & 0 	  & 0\\
0 	  &  0    & 0 	  & 0 \\
\end{pmatrix},
\\
G &= 
\begin{pmatrix}
g_{11} & g_{12} & g_{13} & g_{14}\\
g_{21} & g_{22} & g_{23} & g_{24}\\
g_{31} & g_{32} & g_{33} & g_{34}\\
g_{41} & g_{42} & g_{43} & g_{44}\\
\end{pmatrix}
= 
\begin{pmatrix}
0.95 & 0.03 & 0.0198 & 0.0002\\
0.05 & 0.9  & 0.04  & 0.01\\
0.05 & 0.12 & 0.78  & 0.05\\
0    & 0    &  0    & 1\\
\end{pmatrix}.
\end{align*}

\subsection{LGD model}\label{sec:LGD-model}

We adopt ideas from \citet{frontczak2015modeling} and use the collateral to model the LGD as
\begin{equation}\label{eq:lgdmodel}
\lgd_k = (1-\frac{c_k}{\ead_k})^+,
\end{equation}
where $c_k$ is the value of the collateral at time $k$. At time $k$, the log-return of the collateral value is denoted $\lc_k$, $\lc_k=\frac{c_k}{c_{k-1}}$, $k\geq 1$. The log-return of the collateral is assumed to be modeled by an AR(1) process, see also Equation~\eqref{eq:ARR}, 
\begin{equation*}
\lc_k = 0.73\,\lc_{k-1}+z_k, \, z_k\sim N(0,0.04^2).
\end{equation*}
The value of the collateral, recall \eqref{eq:lgdmodel}, determines the LGD as follows.
\begin{equation*}
\begin{aligned}
\lgd_k &= (1-\frac{c_k}{\ead_k})^+ \\
 &= (1-\frac{c_k}{c_0}*\frac{c_0}{\ead_0}*\frac{\ead_0}{\ead_k})^+\\
 &= (1-\frac{c_k}{c_0}*\frac{1}{\ltv_0}*\frac{1}{n-k+1})^+.
\end{aligned}
\end{equation*}
Note that, in order to determine the $\lgd_k$, the loan-to-value at the horizon $\ltv_0$ should be defined. 
Moreover, for convenience, we further assume that the risk factors for the transition matrices and the log-returns of the collateral factors are independent, hence the default probabilities and the LGD are independent. Note that this assumption is not valid in real life practice. Actually, in practice the correlation between the PD and LGD is a very important parameter in risk calculations. Under the independence assumption the loss in Equation~\eqref{eq:lsimple} is simplified to
\begin{equation}
\begin{aligned}\label{eq:EL_experiment}
L_1^{(r)} & = \sum_{k=2}^n s \mathbb{E}\left[\pd( k\mid \mathcal{F}_k ) \mid \mathcal{F}_h \right] + \mathbb{E}\left[\pd(n \mid \mathcal{F}_n ) \mid \mathcal{F}_h \right] \\ 
& \qquad - \sum_{k=2}^n  \mathbb{E}\left[\left(\pd(k\mid \mathcal{F}_k)-\pd(k-1\mid \mathcal{F}_k)\right) \mid \mathcal{F}_h  \right](1-\mathbb{E}\left[\lgd_k\mid \mathcal{F}_h  \right]).
\end{aligned}
\end{equation}
Given the assumed independence, we only need to approximate 
the multi-year expected probability of default $\mathbb{E}[\pd_{k}^{(r)}]$ by Monte Carlo simulation, while
the expected loss given default  $\mathbb{E}[\lgd_k]$ is analytically computed by using the Black-Scholes formula \eqref{eq:BSF_ELGD} of Corollary~\ref{cor:BSF_ELGD}. 
\medskip\\
In order to assess the performance of the proposed Baysian projection approach, a \emph{waterfall} terminology is used. In the waterfall, the analyses are done step by step as follows.
\begin{enumerate}
    \item The Bayesian and PCA projection approach are compared when used to describe the transition probabilities. The two approaches project the four dimensional factors onto two dimensional factors respectively, and then the transition probabilities based on the two dimensional factors are compared with the original transition probabilities from the four dimensional factors. In particular, for the Bayesian projection approach, Algorithm~\ref{algorithm:projection} is applied. 
    \item  The convergence analysis of the Monte Carlo simulation for expected PD is conducted to find out how many Monte Carlo scenarios are needed to obtain a reliable estimate of the multi-year expected PD.
    \item Then the Bayesian and PCA projection approaches are used in the valuation grid to estimate the multi-year expected PD and accuracies are compared. In this case,  Algorithm~\ref{algorithm:projection2} is applied. 
    \item The closed-form Black-Scholes formula for the expected LGD is checked by comparing with the Monte Carlo estimation. The expected LGDs are used to construct the loss distribution according to Equation~\eqref{eq:lsimple}.
    \item In the end, the proposed Bayesian approach, see Algorithm~\ref{algorithm:projection2}, is used to construct the loss distribution of a unit loan and the risk metrics (expected losses and value-at-risk) are estimated. The accuracy of the estimation is compared with the estimation from the PCA approach.
\end{enumerate}
\noindent
Table \ref{table:waterfall} summaries the waterfalls in the experiments Section~\ref{section:experiment}.  
\begin{table}[h!]
\centering
\begin{tabular}{|l|l|l|}
\hline
Section & Description & Results \\
\hline
Section~\ref{section:experiment_tm} &  Comparison of Bayesian and PCA projection approach & Figure~\ref{fig:error_prj_pca} \\
 & on transition probabilities &\\
\hline
Section~\ref{section:experiment_mc} & Convergence of Monte Carlo simulation for expected PD & Figure~\ref{fig:error_nsim2}\\
\hline
Section~\ref{section:experiment_EPD} & Comparison of Bayesian and PCA projection & Figures~\ref{fig:EPD_prj_pca_1f}, \ref{fig:EPD_prj_pca_2f}\\ 
 & approach on multi-year PD estimation & \\
\hline
Section~\ref{section:experiment_elgd} & Black-Scholes formula for expected LGD & Figure~\ref{fig:ELGD_mc1M_BS} \\
\hline
Section~\ref{section:experiment_el} & Comparison of Bayesian and PCA projection approach & Figures~\ref{fig:EL_prj_pca}-\ref{fig:var999_prj_pca_1f} \\
 & on loss distribution of the target unite bullet loan & \\
\hline
\end{tabular}
\caption{Waterfalls}
\label{table:waterfall}
\end{table}

\subsection{Experiments}\label{section:experiment}

In this section we perform a number of numerical studies to assess and illustrate our projection approach.

\subsubsection{Transition probabilities}\label{section:experiment_tm} 

Suppose the four dimensional model and the two dimensional model are calibrated on the data. In the experiment of this section, we assess the potential error introduced by the dimension reduction approaches, when projecting a higher dimensional transition model, in this case a four factor model, to a lower dimensional transition model, in this case, a two factor model. In particular, we apply the proposed projection approach, see Algorithm~\ref{algorithm:projection} and the PCA approach on the transition probability models and compare the performance of the two approaches. To achieve this, in total 100 different scenarios (risk factors and hence the transition probability matrices) are simulated using the four factor model. Then the proposed projection approach and the PCA approach are used to project the four dimensional factors onto two dimensional factors. For the proposed approach, the two factor model above is used and for the PCA approach the first two principal components are used to match with the number of factors. Based on the two dimensional factors, the transition probabilities can be recomputed according to the two factor model. These transition probabilities from the two factor model are compared with the original transition probabilities based on the four factor model and the relative difference between them are calculated. Figure~\ref{fig:error_prj_pca} shows the relative difference (error) between the transition probabilities of the benchmark (four factor model) and the projections from the proposed approach and PCA approach. One can observe that the differences for the Bayesian projection approach are significantly smaller than those for the PCA approach. 
\begin{figure}[ht]
 \centering
 \includegraphics[scale=0.36]{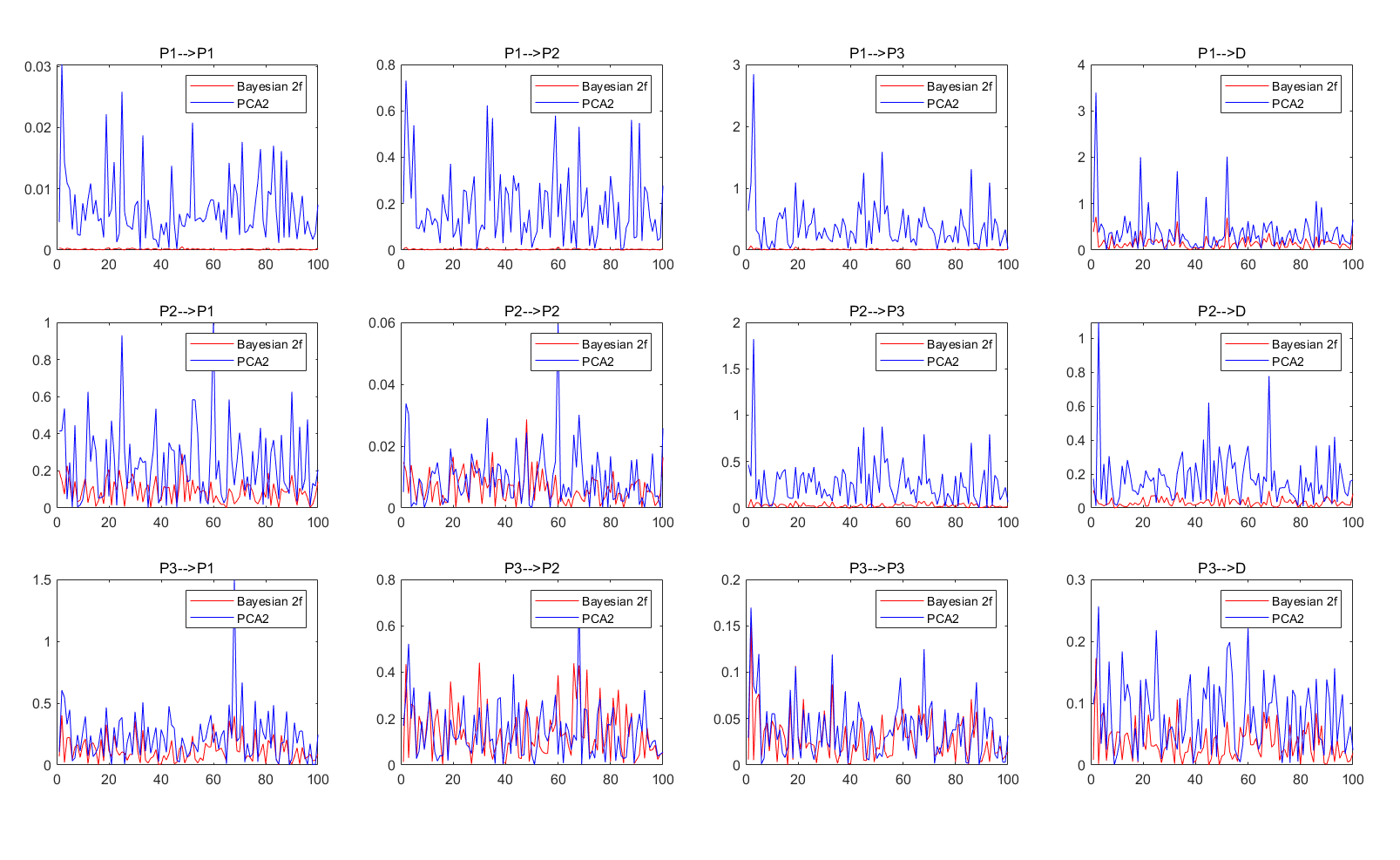}
 \caption{Relative difference of the transition probabilities: the red line is for the proposed Bayesian projection approach and the blue line is for the PCA. On the $x$-axis are the 100 simulated scenarios.}
 \label{fig:error_prj_pca}      
\end{figure}

\subsubsection{Expected PD simulation}\label{section:experiment_mc} 

In order to create the training samples for the expected losses, the expected cumulative PD up to the maturity need to be approximated using Monte Carlo, see Equation~\eqref{eq:EL_experiment}. In the experiment of this section we investigate how many Monte Carlo scenarios are required to obtain a adequate approximation for the expected cumulative PD. For each rating $P1, P2$ and $P3$, we randomly create 1000 starting points for the simulation. These starting points can be treated as the points on the grids. Then starting at each point, the risk factor  paths are simulated up to the maturity, in total $30$ periods. The number of risk factor  paths for each starting points varies from 500 to 20000. The expected cumulative PDs for each starting point are approximated by averaging over the risk factor paths. For the expected cumulative PD approximation 20000 paths are used as the benchmark, while the approximation using 500, 1000, 5000 and 10000 paths are used as the test sets. The relative difference of the expected cumulative PD between the benchmark and the test sets are computed and the results are presented in Figure~\ref{fig:error_nsim2}. The results indicate that the expected cumulative PD can be adequately approximated by using 1000 paths. 

\begin{figure}[ht]
 \centering
 \includegraphics[width=0.95\textwidth]{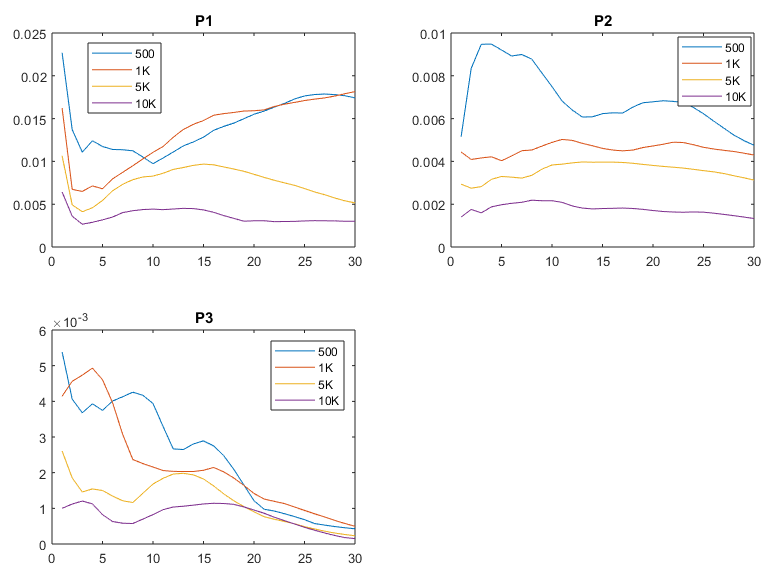}
 \caption{Maximum relative difference of the expected cumulative PD between 20000 MC paths and 500, 1000, 5000, 10000 paths: the $y$-axis is the maximum relative difference of the expected cumulative PD for the 1000 different starting points, on the $x$-axis are the time periods up to 30.}
 \label{fig:error_nsim2}      
\end{figure}

\subsubsection{Experiments for the projection approach for expected PD}
\label{section:experiment_EPD}
The purpose of the experiments is to assess the performance of the proposed Bayesian filter dimension reduction approach on valuating the losses, see Algorithm~\ref{algorithm:projection2}. Due to the independent assumption between the PD and LGD, the loss valuation is simplified such that  the interpolation grid is only required for the expected PD valuation, see Equation~\eqref{eq:lsimple}. While the expected LGD can be easily computed by using the Black- Scholes formula in Appendix~\ref{appendix:BSF_LGD}. Therefore, in this section, we apply Algorithm~\ref{algorithm:projection2} to compute the expected PD and the performance of the proposed Bayesian projection approach is assessed by comparing with the PCA approach. The assessment for the loss valuation will be presented later in Section~\ref{section:experiment_el}. \\
\\
Based on the four and two dimensional transition models, one can train the interpolation grid by simulating the expected PD for each point on the grid. As we described in Algorithm~\ref{algorithm:projection2}, the simulation of the expected PD uses the higher dimensional transition model while the grid points are in a lower dimensional space,  obtained by a dimension reduction approach, i.e.\ the proposed Bayesian projection approach or the PCA approach. Then the valuation grid can be trained using the simulated expected PD and the projected factors. We compare the accuracy of the valuation grid based on the proposed Bayesian filter projection approach and the PCA approach. 
To build the valuation grid, the expected PD are simulated from the four factor model specified in Section~\ref{section:tm}. 

The starting points of the simulations (i.e.\ the higher dimensional factors) consist of two parts: one contains the pre-determined points to cover the possible extreme simulated values of the factors, the other contains random  simulated points to fill the density. In this experiments, the pre-determined points are chosen linearly spaced from minus four sigma to plus four sigma, i.e.\ for the $i$-th component of the factor, the linear spaced points are chosen from interval $[-4\sigma_i, 4\sigma_i]$, with $\sigma_i$ the standard deviation of the $i$-th component of the factor. For each dimension, we chose $15$ points. Moreover, $1000$ more points are randomly simulated from the joint distribution of the four dimensional factors. Therefore, in total there are $15^4 + 1000 = 51625$ points in the valuation grid. Starting from the factors, each term structure of the expected PD is estimated by Monte Carlo simulation with $1000$ paths. The term structure is considered from $1$ to $30$ periods. Meanwhile, the four dimensional factors are projected to lower dimensional factors using the Bayesian filter projection approach and the PCA approach. The chosen lower dimensional models are the one-factor and two-factor model. Once the valuation grids are trained, $1000$ new four dimensional factors are simulated as the test set to assess the performance of different valuation grids based on different approaches. Based on these $1000$ new four dimensional factors, the term structures of the expected PD are estimated by using 20000 Monte Carlo paths. These Monte Carlo estimations are used as the benchmark for the comparison.  

Figures~\ref{fig:EPD_prj_pca_1f} and~\ref{fig:EPD_prj_pca_2f} show the average (over the 20000 MC paths) relative difference on the term structure between the valuation grids and the Monte Carlo benchmarks. Specially, Figure~\ref{fig:EPD_prj_pca_1f} presents the comparison between the Bayesian filter projection using one factor model and the PCA projection using the first principal component, while Figure~\ref{fig:EPD_prj_pca_2f} presents the Bayesian filter projection of the two factor model and the PCA projection using the first two principal components. One observes that for one factor projection, in general PCA approach performs slightly better than the Bayesian filter projection approach. In particular, the Bayesian filter projection performs better for the first rating, but worse for rating 2 and rating 3. This is because the one factor model mainly focuses on the transitions of the first rating and hence the errors for the other two ratings are large. For the two factor projection, on the other hand, the Bayesian filter projection performs much better than the PCA approach. Comparing Figures~\ref{fig:EPD_prj_pca_1f} and \ref{fig:EPD_prj_pca_2f}, one observes very big improvements on the relative error for the Bayesian filter projection, while the improvement for the PCA projection is limited.

\begin{figure}[ht]
 \centering
 \includegraphics[width=0.95\textwidth]{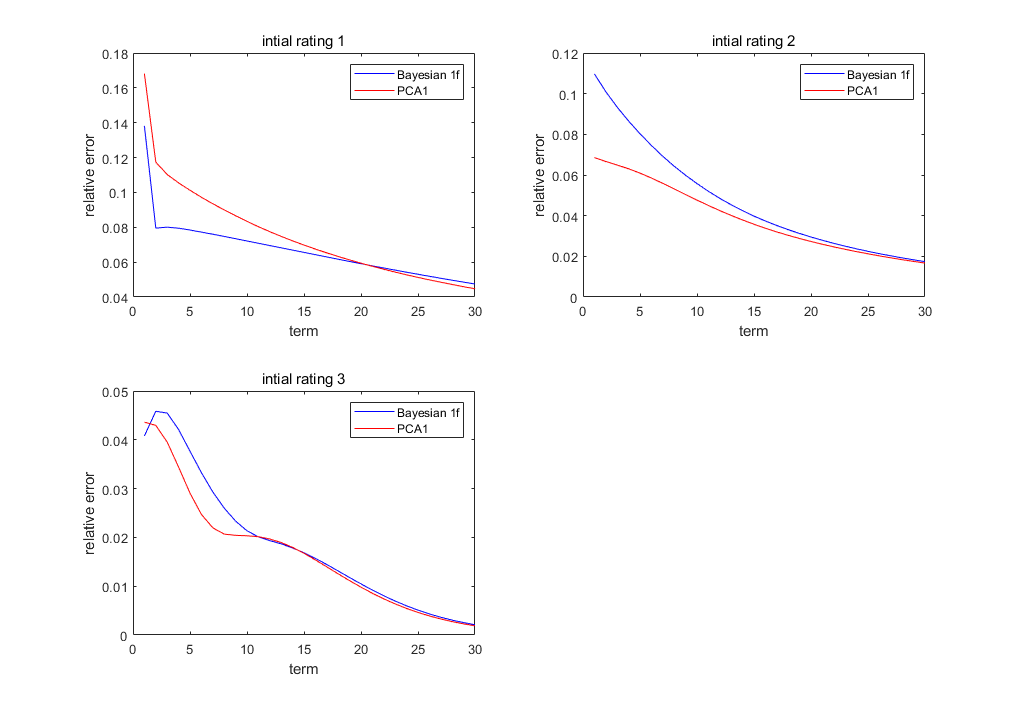}
 \caption{Average relative error on expected PD: Bayesian filter projection vs.\ PCA projection, one factor.}
 \label{fig:EPD_prj_pca_1f}      
\end{figure}

\begin{figure}[ht]
 \centering
 \includegraphics[width=0.95\textwidth]{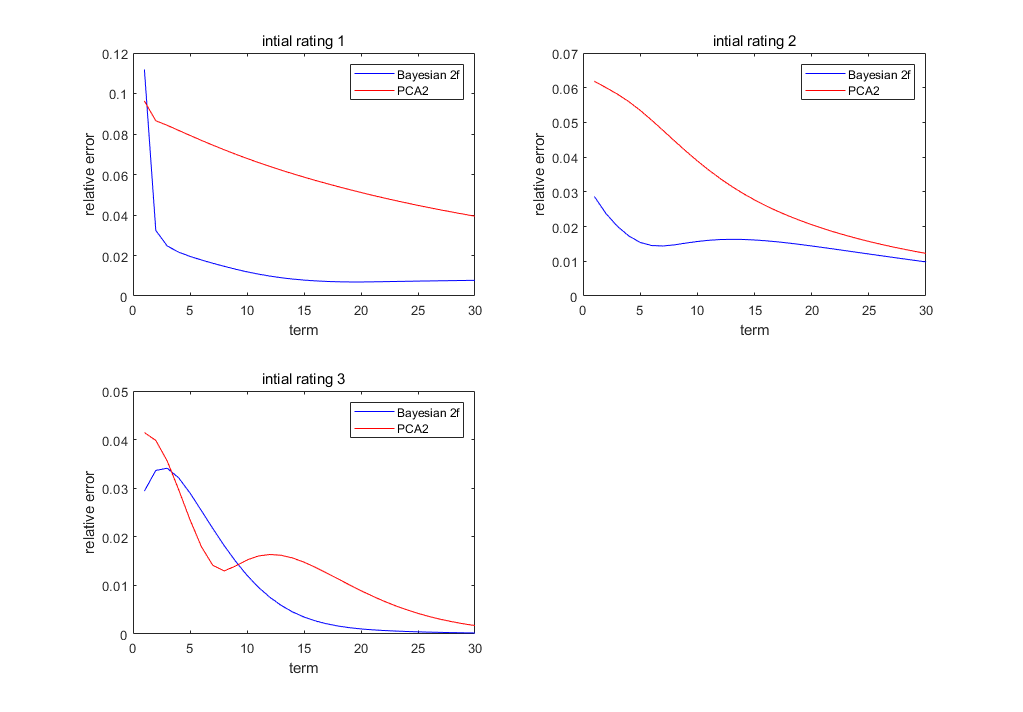}
 \caption{Average relative error on expected PD: Bayesian filter projection vs.\ PCA projection, two factors.}
 \label{fig:EPD_prj_pca_2f}      
\end{figure}

\subsubsection{Expected LGD: Black-Scholes formula v.s. Monte Carlo}\label{section:experiment_elgd} 

Before proceeding to the experiment on the loss valuation, in the experiment of this section, we illustrate the Black Scholes formula, the analytical solution of Equation~\eqref{eq:BSF_ELGD} in Corollary~\ref{cor:BSF_ELGD}, for the expected LGD (ELGD), by comparing it to the Monte Carlo approximation. Note that the ELGD depends on the initial LTV and the initial log-return of the collateral. In this case, we choose the initial LTV to be 0.8, 1, 1.5 and 2, the initial log-return of the collateral to be -0.1, 0, and 0.1. Figure~\ref{fig:ELGD_mc1M_BS} shows the term structure of the ELGD computed by the Black-Scholes formula and the Monte Carlo simulation with 1 million scenarios. One can easily see that the outcomes of the two approaches align very well with each other. 
\begin{figure}[ht]
 \centering
 \includegraphics[width=0.95\textwidth]{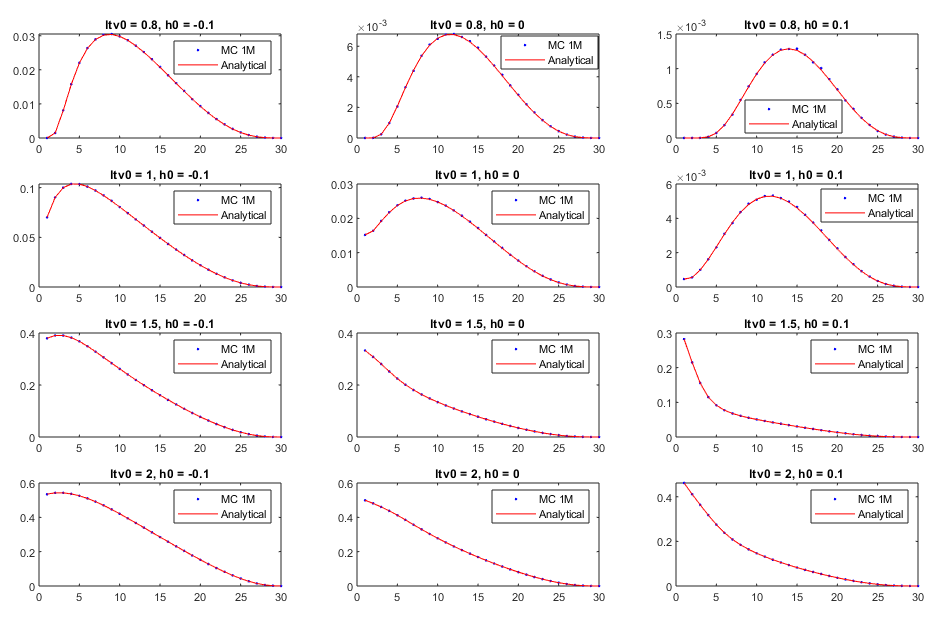}\caption{Term structure of ELGD: Black-Scholes (analytical) vs.\ Monte Carlo.}
 \label{fig:ELGD_mc1M_BS}      
\end{figure}

\subsubsection{Experiments for the projection approach for risk measures of the loss distribution}\label{section:experiment_el}

These experiments assess the accuracy of the proposed projection approach when estimating the quantiles of the loss distribution. In order to obtain a deeper insight about the performance, we investigate the loss distribution per initial rating, instead of the total loss distribution, which aggregates over different ratings. In this way we can compare the performance of the projection approaches for different initial ratings. To be more specific, the loss distribution per initial rating is constructed as follows.
\begin{enumerate}
\item Simulate one million scenarios (risk factors and collateral process) using the four factor model,
\item project the simulated risk factors into lower dimensional risk factors,
\item for each initial rating, given the lower dimensional factors, estimate the expected PD per scenario using the trained grid as in Section~\ref{section:experiment_EPD},
\item given the simulated collateral value, compute the expected LGD using the Black Scholes formula \eqref{eq:BSF_ELGD}.
\item Note that the PD and LGD are assumed to be independent, therefore the loss per scenario can be evaluated, according to equation \eqref{eq:Vr}, by using the estimated expected PD and the expected LGD.
\end{enumerate}
To illustrate the performance, the loss distribution based on the four factor model is used as a benchmark. In this benchmark, instead of using the projection approach and the interpolation grids, the expected PD per scenario in step 3 mentioned above is estimated directly by Monte Carlo simulation based on the four dimensional model itself.  The chosen risk measures are expected loss, 95\%-, 99\%- and 99.9\%-quantiles. We compute the absolute relative difference of these risk measures between the estimates from  projection approaches and from the benchmark. Figures \ref{fig:EL_prj_pca}--\ref{fig:var999_prj_pca_1f} present the comparison of the relative differences between different projection approaches for different risk measures: expected loss (EL), 95\%-, 99\%- and 99.9\%- quantiles, respectively. In each figure, the dotted lines and the solid lines show the relative difference for the Bayesian projection approach and for PCA approaches, respectively, while the blue and red lines present the relative difference of one factor and two factor models (PCA), respectively. 
\medskip\\
In general, the results are in line with the results in Section~\ref{section:experiment_EPD}. The Bayesian projection approach with the two factor model performs significantly better than the other candidates. Specially it has significantly smaller error than the PCA approach with the first two principal components. By contrast, the Bayesian projection approach with the one factor model does not outperform the PCA approach with the first principal component. The reason is that the Bayesian projection approach with the one factor model focuses on the best rating, i.e.\ rating~1, which has the majority number of clients, while the PCA approach does not have this concentration. Therefore, one observes a better performance of one-factor Bayesian projection in the best rating but worse performance in rating 2 and 3. Another important observation is that, the Bayesian projection approach has a much better performance in the best rating than the PCA approach when estimating the more extreme quantile of the loss distribution. in particular one observes in Figure~\ref{fig:var999_prj_pca_1f} that even the one factor Bayesian projection has comparable performance as the two-factor PCA projection in the best rating when estimating the 99.9\%-quantile. This is a desired property since in practice most of the exposures are located in the low-default-probability portfolios.

\begin{figure}[ht]
 \centering
  \includegraphics[width=0.95\textwidth]{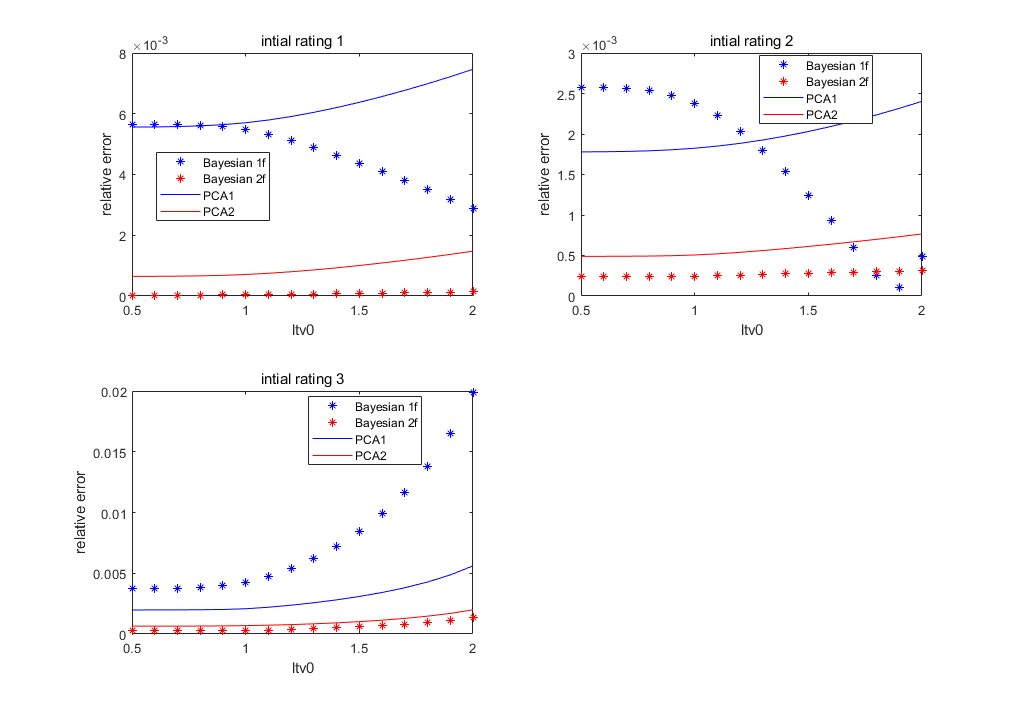}
\caption{Absolute relative difference on expected loss: Bayesian filter projection vs.\ PCA projection.}
 \label{fig:EL_prj_pca}      
\end{figure}

\begin{figure}[ht]
 \centering
 \includegraphics[width=0.95\textwidth]{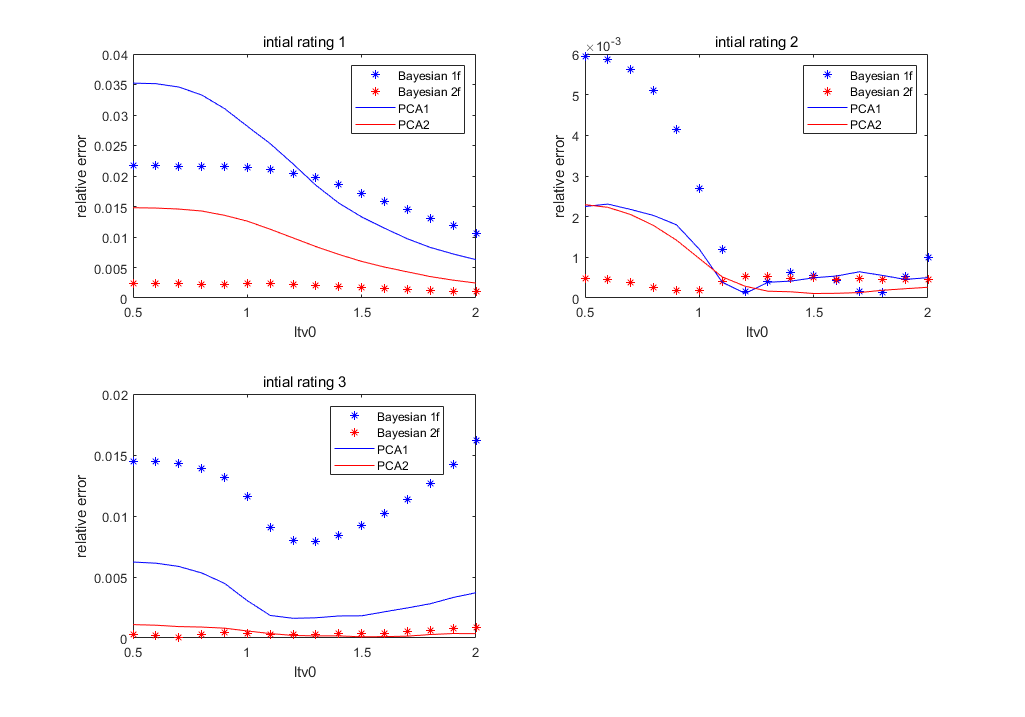}
 \caption{Absolute relative difference on 95\%-quantile: Bayesian filter projection vs.\ PCA projection.}
 \label{fig:var95_prj_pca}      
\end{figure}

\begin{figure}[ht]
 \centering
 \includegraphics[width=0.95\textwidth]{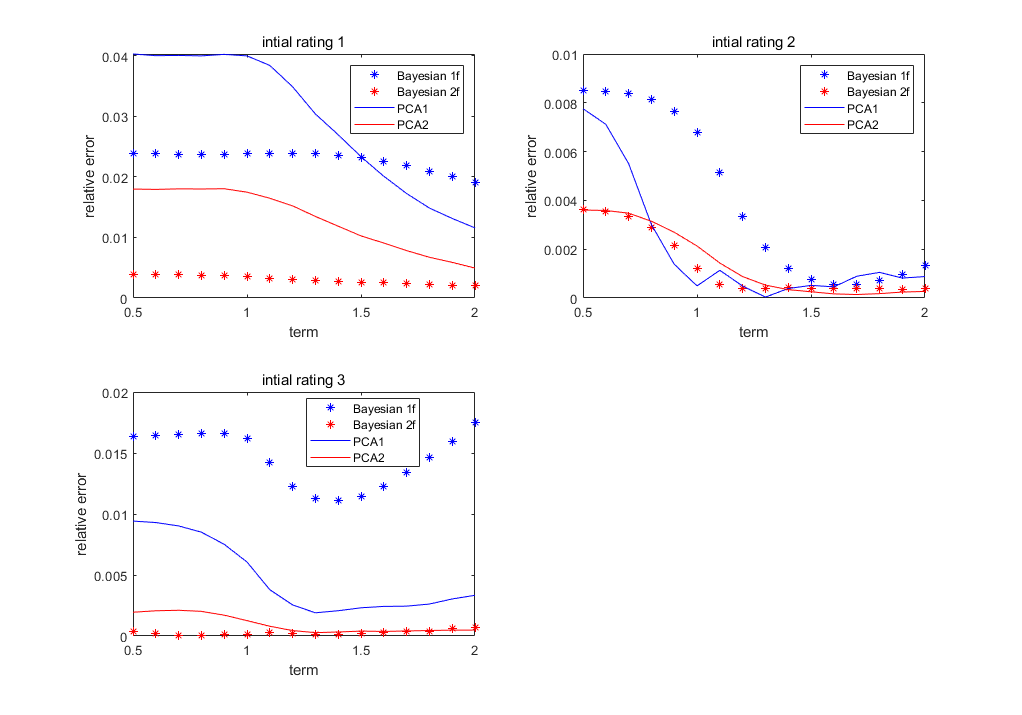}
 \caption{Absolute relative difference on 99\%-quantile: Bayesian filter projection vs.\ PCA projection.}
 \label{fig:var99_prj_pca_1f}      
\end{figure}

\begin{figure}[ht]
 \centering
 \includegraphics[width=0.95\textwidth]{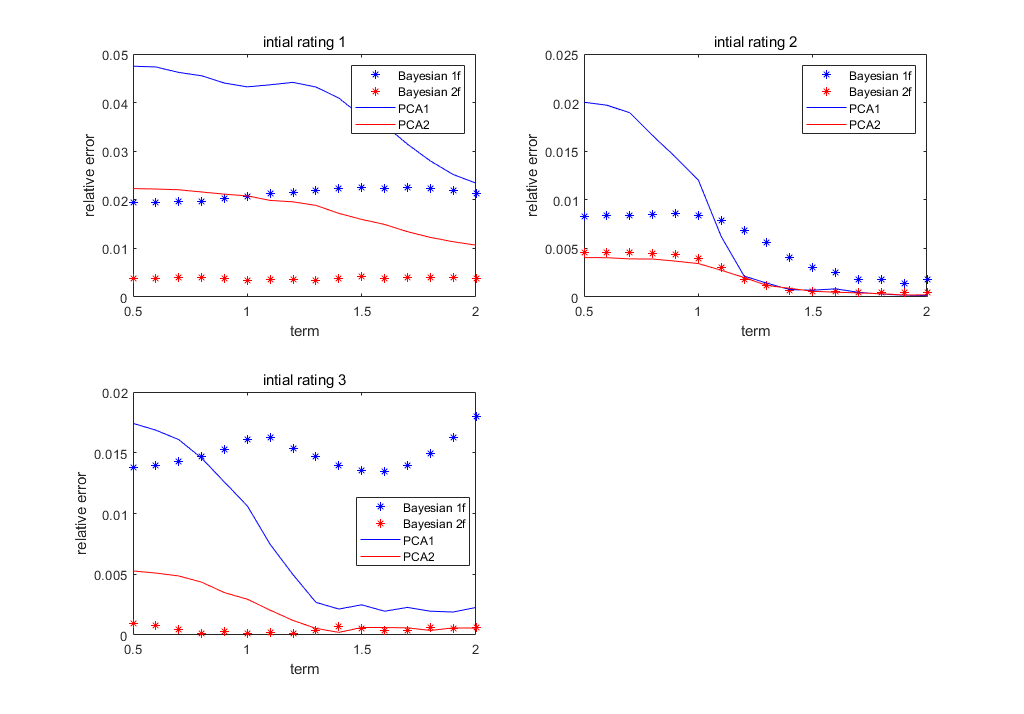}
 \caption{Absolute relative difference on 99.9\%-quantile: Bayesian filter projection vs.\ PCA projection}
 \label{fig:var999_prj_pca_1f}      
\end{figure}

\section{Conclusions}\label{section:conclusions}

In this paper we developed a dimension reduction methodology based on the Bayesian filter and smoother. This methodology is designed to achieve a fast and accurate loss valuation algorithm in the credit risk modelling, but it can also be extended to the valuation models of other risk types. The proposed methodology is generic, robust and can easily be implemented. Moreover, we showed the accuracy of the proposed methodology in the estimation of expected loss and value-at-risk by numerical experiments. The results suggest that, compared to the currently most used PCA approach, the proposed methodology provides more accurate estimation on the expected loss and value-at-risk of a loss distribution.

\clearpage
\appendix

\section{Analytic formula for expected loss-given-default}\label{appendix:BSF_LGD}

The LGD at time $k$ conditional on the initial information, the sigma-algebra $\mathcal{F}_0$, is defined by
\begin{equation*}
\elgd_k = \mathbb{E}\left[\lgd_k | \mathcal{F}_0 \right].
\end{equation*}
In the present paper, we assume the LGD depends on the collateral value, see \eqref{eq:lgdmodel}. Hence we have
\begin{equation*}
\elgd_k = \mathbb{E}\left[(1-\frac{c_k}{\ead_k})^+ | \mathcal{F}_0 \right],
\end{equation*}
with $c_k$ the value of collateral and $\ead_k$ the exposure-at-default at time $k$. Since we assume the $\ead_k$ is deterministic, the $\elgd_k$ is determined by the stochastic nature of the collateral value $c_k$.  The loan-to-value ratio is $\ltv_0 = \frac{\ead_0}{c_0}$, hence we have
\begin{align}
\elgd_k &= \mathbb{E}\left[(1-\frac{c_k}{c_0}\frac{c_0}{\ead_0}\frac{\ead_0}{\ead_k})^+ | \mathcal{F}_0 \right]\nonumber\\
&= \mathbb{E}\left[(1-\frac{c_k}{c_0}\frac{1}{\ltv_0}\frac{\ead_0}{\ead_k})^+ | \mathcal{F}_0 \right]\nonumber\\
&= \mathbb{E}\left[(1-K_k\frac{c_k}{c_0})^+ | \mathcal{F}_0 \right],\label{eq:elgdz}
\end{align}
with $K_k = \frac{1}{\ltv_0}\frac{\ead_0}{\ead_k}$. We suppose that the log-returns $\lc_k = \log\frac{c_k}{c_{k-1}}$ of the collateral follow an AR(1) process,
\begin{equation}\label{eq:ARR}
\lc_k = c + x \lc_{k-1} + \sigma z_k, \, z_k\sim N(0,1),\, k\geq 1,
\end{equation}
with $|x|<1$. Define the numbers $a_k$ by the iteration
\begin{equation}\label{eq:at}
a_k = 1 + x a_{k-1},\, a_0=0,
\end{equation}
which results in the analytic expression $a_k = \frac{1-x^k}{1-x}$. We now have the following result.
\begin{prop}\label{prop:collateral}
For given $\lc_0$, it holds that
\begin{equation}
\label{eq:collateral}
\log\frac{c_t}{c_0} \sim N(\mu_t, \Omega_t),
\end{equation}
with
\begin{align*}
\mu_t &= c\sum_{i=1}^t a_i + a_tx\, \lc_0,\\
\Omega_t &= \sigma^2\sum_{i=1}^t a_i^2.
\end{align*}
\end{prop}

\begin{proof}
It follows from the recursion \eqref{eq:ARR} that all $\lc_k$ are normal (conditional on $\lc_0$). Hence for \eqref{eq:collateral} to hold we only have to find mean and variance. We first prove a general result. Namely,  for $s=0,\ldots, k$, the equation
\begin{equation}
\label{eq:loop}
\log\frac{c_k}{c_0} + \lc_0 = c\sum_{i=1}^s a_i + a_{s+1}\lc_{k-s} + \sum_{i=0}^{k-s-1} \lc_i + \sigma\sum_{i=1}^{s}a_i z_{k-i+1}
\end{equation}
holds. We now show \eqref{eq:loop} by induction w.r.t.\ the variable $s=0,\ldots,k$.
When $s = 0$,
\begin{equation*}
\log\frac{c_k}{c_0} = \sum_{i=1}^k \lc_i,
\end{equation*}
which is Equation~\eqref{eq:loop} for $s=0$ by definition of the $\lc_i$. Suppose then that Equation~\eqref{eq:loop} holds for $s=t$. Using \eqref{eq:ARR} with $k$ replaced by $k-t$ and \eqref{eq:at} with $k$ replaced by $t+2$, one gets starting from the induction assumption
\begin{equation*}
\begin{aligned}
\log\frac{c_k}{c_0} + \lc_0 &= c\sum_{i=1}^t a_i + a_{t+1}(c + x \lc_{k-t-1} + \sigma z_{k-t}) + \sum_{i=0}^{k-t-1} \lc_i +  \sigma\sum_{i=1}^{t}a_i z_{k-i+1}\\
& = c\sum_{i=1}^{t+1} a_i + a_{t+1}x \lc_{k-t-1} + \lc_{k-t-1} + \sum_{i=0}^{k-t-2} \lc_i + \sigma\sum_{i=1}^{t+1}a_i z_{t-i+1}\\
& = c\sum_{i=1}^{t+1} a_i + a_{t+2}\lc_{k-t-1} + \sum_{i=0}^{k-t-2} \lc_i + \sigma\sum_{i=1}^{t+1}a_i z_{t-i+1},
\end{aligned}
\end{equation*}
which is indeed Equation~\eqref{eq:loop} for $s=t+1$, as was to be shown. 

Taking $k = t$ and $s=k$ in Equation~\eqref{eq:loop}, we have, using \eqref{eq:at} once more,
\begin{align*}
\log\frac{c_t}{c_0} & = -\lc_0 + c\sum_{i=1}^{t} a_i + a_{t+1}\lc_{0} + \sigma\sum_{i=1}^{t}a_i z_{t-i+1} \\
& = c\sum_{i=1}^{t} a_i + xa_t\lc_{0} + \sigma\sum_{i=1}^{t}a_i z_{t-i+1}.
\end{align*}
Therefore we obtain the distribution of $\log\frac{c_t}{c_0}$, conditional on $ \lc_{0}$,
\begin{equation*}
\log\frac{c_t}{c_0}\mid \lc_{0}\sim N(\mu_t, \Omega_t)
\end{equation*}
with $\mu_t = c\sum_{i=1}^{t} a_i + x a_t \lc_0$ and $\Omega_t = \sigma^2\sum_{i=1}^{t}a_i^2$, which finishes the proof.
\end{proof}
\noindent
As a corollary we obtain a Black-Scholes like formula for the expected loss given default.
\begin{cor}\label{cor:BSF_ELGD}
For the expected loss given default one has
\begin{equation}
\elgd_t = \Phi(-d_2) - K_te^{\mu_t+\frac{\Omega_t}{2}}\Phi(-d_1),
\label{eq:BSF_ELGD}
\end{equation}
with $\Phi$ the standard normal cumulative probability function, and
\begin{equation*}
\begin{aligned}
d_1 &= \frac{\log(K_t) + \mu_t + \Omega_t}{\sqrt{\Omega_t}},\\
d_2 &= \frac{\log(K_t) + \mu_t}{\sqrt{\Omega_t}}.
\end{aligned}
\end{equation*} 
\end{cor}

\begin{proof}
The proof is completely analogous to the proof of the classical Black-Scholes formula by realizing that Equation~\eqref{eq:elgdz} gives $\elgd_t=\mathbb{E}\left[(1-K_t\log Z)^+ | \mathcal{F}_0 \right]$, where $Z$ has the conditional normal distribution with mean $\mu_t$ and variance $\Omega_t$ as in  Proposition~\ref{prop:collateral}.
\end{proof}

\bibliographystyle{plainnat} 
\bibliography{A_dimension_reduction_approach_for_loss_valuation_in_credit_risk_modelling}

\end{document}